\documentclass[11pt]{article}
\usepackage{amsmath}
\usepackage{amssymb}
\usepackage{amsthm}
\usepackage{anysize}
\usepackage{natbib}
\usepackage{pdflscape}
\usepackage{dsfont}
\usepackage{enumerate}
\usepackage{graphicx}
\usepackage{floatrow}

\usepackage{multicol}
\usepackage{multirow}
\usepackage[dvips]{epsfig}

\usepackage{float}
\usepackage{color}
\usepackage[latin1]{inputenc}
\usepackage[english]{babel}
\usepackage{xcolor}

\usepackage{bbm}
\usepackage{mathtools}
\numberwithin{equation}{section}

\theoremstyle{plain}
\newtheorem{thm}{Theorem}[section]
\newtheorem{cor}[thm]{Corollary}
\newtheorem{lem}[thm]{Lemma}
\newtheorem{prop}[thm]{Proposition}
\theoremstyle{definition}
\newtheorem{defn}[thm]{Definition}
\theoremstyle{definition}

\theoremstyle{remark}
\newtheorem{oss}{Remark}[section]
\newtheorem{esmp}{Example}[section]

\marginsize{25mm}{25mm}{25mm}{25mm}

\newcommand{\ii}{\mathrm{i}}
\newcommand{\I}{1{\hskip -2.5 pt}\hbox{I}}
\newcommand{\eps}{\varepsilon}

\newcommand{\CC}{\mathbb{C}}
\newcommand{\NN}{\mathbb{N}}
\renewcommand{\P}{\mathbb{P}}
\newcommand{\QQ}{\mathbb{Q}}
\newcommand{\Oo}{\mathcal{O}}
\newcommand{\Bb}{\mathcal{B}}
\newcommand{\Ll}{\mathcal{L}}
\newcommand{\Nn}{\mathcal{N}}
\newcommand{\R}{\mathbb R}
\newcommand{\E}{\mathbb E}
\newcommand{\VV}{\mathbb{V}}
\newcommand{\BS}{\mathrm{BS}}
\newcommand{\Cov}{\mathrm{Cov}}

\renewcommand{\S}{\mathcal S}
\newcommand{\D}{\mathcal D}
\newcommand{\YS}{Y^{\scriptsize SL}}

\newcommand{\hyp}{{}_1F_1}
\newcommand{\Kurt}{\mathrm{Kurt}}
\newcommand{\Skew}{\mathrm{Skew}}

{\it}{\rm}
{\it}{\rm}
{\it}{\rm}

\DeclareMathOperator*{\argmin}{arg\,min}
\begin{document}
\title{
Anomalous diffusions in option prices: connecting trade duration and the volatility term structure\footnote{The authors would like to thank Mark Meerschaert and Peter Straka for the valuable suggestions.}}
\author{Antoine Jacquier\footnote{Department of Mathematics, Imperial College London, and Alan Turing Institute.} \hspace{2cm} Lorenzo Torricelli\footnote{Department of Economics and Management. Email: lorenzo.torricelli@unipr.it.} }
\date{\today}

\maketitle

\begin{abstract}
Anomalous diffusions arise as scaling limits of  continuous-time random walks (CTRWs) whose innovation times are distributed according to a power law. The impact of a non-exponential waiting time does not vanish with time and leads to different distribution spread rates compared to standard models. In financial modelling this has been used to accommodate for random trade duration in the tick-by-tick price process. 
We show here that anomalous diffusions are able to reproduce the market behaviour of the implied volatility more consistently than usual L\'evy or stochastic volatility models.
Two distinct classes of underlying asset models are analyzed: 
one with independent price innovations and waiting times, 
and one allowing dependence between these two components. 
These models capture the well-known paradigm according to which shorter trade duration is associated with higher return impact of individual trades. We fully describe these processes in a semimartingale setting leading to no-arbitrage pricing formulae,  study their statistical properties, and in particular observe that skewness and kurtosis of asset returns do not tend to zero as time goes by. 
We finally characterize the large-maturity asymptotics of Call option prices, 
and find that the convergence rate to the spot price is slower than in standard L\'evy regimes, 
which in turn yields a declining implied volatility term structure and a slower time decay of the skew.
\end{abstract}


\section{Introduction}
In quantitative finance, models of asset returns typically evolve according to It\^o diffusions or L\'evy-type models. 
From a microstructural point of view, these can be seen as scaling limits of continuous-time random walks (CTRWs) with exponentially distributed inter-arrival times. Instead, time changing CTRWs to a renewal process whose waiting times obey a power law yields, in the scaling limit, an \emph{anomalous diffusion}, 
namely a space-time propagation process where the particle spreads at a rate 
different from linear, which is observed in the classical diffusive case. 
The use of anomalous diffusions in financial models was pioneered by~\cite{CTRWfinance2} and~\cite{CTRWfinance}, and they have proved useful to capture memory effects, trade idle time, 
and other microstructural price features exhibited by high-frequency time series.  

However, applications of anomalous diffusions for continuous-time option pricing have so far been scarce. 
The sub-diffusive Black-Scholes model was introduced in~\cite{SubBS} to capture asset staleness and periods of trade inactivity, but implications on option pricing and implied volatilities were not illustrated. 
\cite{CarteaMB} analysed the volatility surface of a CTRW whose innovation times are distributed according to a Mittag-Leffler hazard function, produced explicit option pricing formulae, 
and provided evidence that the long-term skewness and smile can be captured.

We  show here how anomalous diffusions in equity returns can also capture the long-term 
behaviour of the implied volatility surface. Specifically, we argue that the persistence of a slowly decaying volatility skew can be explained by postulating the survival of trade durations effects at longer maturities. 
We consider returns and innovation time random walks which converge in the scaling limit to a pair of L\'evy processes, one of which is a subordinator. 
According to~\cite{CTRWCoupledLlimits, triangular, triangularErrata, StrakaHenry, MeerscahertOCTRW}, 
the associated CTRW time-changed with the renewal process of the innovation times converges to an anomalous diffusion which can be represented as a time-changed L\'evy process. 
One appealing feature is that both analytical formulae for the Laplace transforms (in the time variable) 
of the characteristic function of this limit and integral expressions for the density functions 
(in terms of the L\'evy measures) are known .

We analyse two distinct classes of anomalous diffusion models.  
The first is the purely \emph{subdiffusive L\'evy model}  (SL), where the CTRW limiting diffusion consists of a L\'evy process time-changed by an independent inverse-stable subordinator. 
Several instances of such models have been already investigated in the literature. In terms of the generating fractional Fokker-Planck equations such a class has been investigated  in~\cite{CarteDCNegrete}. 
The particular case where the parent L\'evy process is a Brownian motion was introduced in~\cite{SubBS}; the compound Poisson  case in~\cite{CarteaMB}. A general treatment when the driving noise is a generic L\'evy process has been recently provided in \cite{TorricelliTLS}.
Moreover the classical models in~\cite{CTRWfinance2},~\cite{CTRWfinance} 
also admit a representation of the SL form: we revisit such models as stochastic time changes, well suited tools for option pricing purposes. 

The time change representation of subdiffusive models also paves the way for our second second class of models, developing an idea that originally appeared in~\cite[Example 2.8]{CTRWCoupledLlimits} and \cite[Examples 5.4]{MeerscahertOCTRW}. This asset price evolution realistically incorporates the dependence between the L\'evy parent returns generating process and the inverse-stable subordinator modelling the trades waiting time. 
We call it the model with \emph{dependent returns and trade duration} (DRD).

Apart from being natural outcomes of time-changed random walk tick-by-tick price models, 
these two models are strongly  supported by the econometric  analysis by~\cite{Engle} and~\cite{EngleDufour}, and confirmed in numerous empirical studies later on. 
The evidence is that trading activity is inversely correlated with price impact, i.e. the `volatility' of the asset price: the fewer the trades (longer duration), the more sluggish the price innovations; 
conversely, intense trading (short duration) is associated with higher price excursions. 
Remarkably, this principle is captured in our setting.

We describe such equity models in a semimartingale dynamic setting, leading to no-arbitrage pricing relations under appropriate equivalent risk neutral measures. 
Using the results of~\cite{MeerscahertOCTRW} on the Fourier-Laplace transforms of anomalous diffusions, 
we  further provide familiar Parseval-Plancherel formulae for option prices in the spirit of~\cite{Lewis}. Additionally, we study the moments and serial correlation properties of the model and show that skewness and kurtosis of the asset returns in the DRD model converge for large times, 
and do not vanish, contrary to L\'evy models, leading in particular to profound differences on the long-term volatility smile.

Finally we characterize the large-maturity behaviour of Call options and find that the convergence rate is much slower than in standard L\'evy or stochastic volatility regimes. 
We uncover a relationship according to which a declining implied volatility level implies a slowly decaying skew,
at least compared to that of L\'evy and exponentially affine models. 
But we find that a (slowly) vanishing volatility level is a defining feature of these models, 
due to long-maturity prices converging much slower than in standard models. 
Ultimately, for the DRD model we show that the vanishing rate of the skew is slower than the usual $1/T$, 
in line with market data.
As illustrated in the calibration in Section~\ref{numerical}, 
the practical importance of anomalous diffusion model is that the `duration parameter'~$\beta$ 
improves the cross-sectional fit to multiple maturities compared to a L\'evy model, while having virtually no impact on the short-maturity calibration. 
This justifies the interpretation of~$\beta$ as a long~term skew component.

We believe the contribution of this work to be manifold. 
We establish an explicit structural connection between trade duration and skew persistence; 
we introduce an analytical model that accounts for trades duration and dependence between trade waiting times and returns, consistent with the Econometrics literature; 
we systematically unify the treatment of SL models under the umbrella of a single time-changed representation and the corresponding analytic pricing formulae; 
finally we extend the analysis of the `Beta-time' process in~\cite{CTRWS} and~\cite{MeerscahertOCTRW}, 
providing its moments and statistical properties through its time-changed representation.

In Section~\ref{notation} we introduce fundamental building blocks and some useful notations. 
In Section~\ref{microstructure} we introduce the CTRWs components of the base tick-by-tick model and the convergence theorem leading to their limiting continuous-time versions. 
The anomalous diffusions are introduced in Section~\ref{model}, together with their analytical properties and time-changed semimartingale representations, while their statistical properties are characterized in Section~\ref{statistical}.
In Section~\ref{price} we show how to construct equivalent pricing measures, and provide an integral price representation for European Call option prices.
This allows us to study in Section~\ref{surface} the structure of the corresponding implied volatility, 
with a particular emphasis on its large-maturity properties. 
Finally in Section~\ref{numerical}, we numerically highlight interesting features of the SL and DRD models, 
and show that both models allow for a good fit to market data.



\section{Foundational elements}\label{notation}
We follow here~\cite[Chapter 1]{Kyprianou}.
In a market filtration $(\Omega, \mathcal{F} ,(\mathcal F_t)_{t\geq 0}, \P)$,
 a L\'evy process~$X$ is uniquely characterized by its L\'evy exponent,  namely the function~$\psi_X:\CC\to\CC$ defined via the relation
$\E\left[e^{-\ii z X_t}\right]=\exp\left(-t \psi_X(z)\right)$, 
and given explicitly by the L\'evy-Khintchine formula
\begin{equation}\label{eq:LevyKhint}
\psi_X(z) = \ii z \mu  + \frac{ z^2 \sigma^2  }{ 2 } - \int_\mathbb R (e^{-\ii z x}-1 + \ii z x\I_{|x|<1})\nu(dx),
\end{equation}
where $\mu\in\R$, $\sigma\geq 0$, and~$\nu$ is a measure concentrated on $\R\setminus\{0\}$ such that
$\int_{\R}(1\wedge x^2)\nu(dx)$ is finite. In order to guarantee some minimal properties of the asset pricing model (existence of the first moment), we always assume that
\begin{equation}\label{firstfinitemoment}
\int_{|x|>1} e^{x}\nu(dx) <\infty.
\end{equation}

A subordinator~$L$ is an almost surely non-decreasing L\'evy process,
with L\'evy measure~$\nu_L$ supported on~$(0,\infty)$,
and the L\'evy-Khintchine representation for its Laplace exponent 
defined via the relation $\E\left[e^{-s X_t}\right]=\exp\left(-t \phi_L(s)\right)$ simplifies to
\begin{equation}
\phi_L(s)= s \mu   - \int_0^{\infty}(e^{- s u}-1) \nu_{L}(du),
\end{equation}
for $\mu>0$, and where $\int_0^{\infty} u\nu_{L}(du) < \infty$.
A bivariate L\'evy process $(X,L)$, with~$L$ a subordinator, 
has joint Fourier-Laplace transform 
$\E[e^{-\ii z X_t - s L_t}]=\exp\left(-t\psi_{X,L}(z,s)\right)$
of the form
\begin{equation}
\psi_{X,L}(z,s)=\ii z \mu_X + s \mu_L + \frac{ z^2 \sigma^2  }{ 2 } - \int_{\R}\int_0^{\infty}\left(e^{-\ii z x - s u}-1 + \ii z x\I_{|x|<1}\right)\nu_{X,L}(dx, du),
\end{equation}
with L\'evy-Laplace triplet $((\mu_X, \mu_T), \sigma,  \nu_{X,T})$. 

For a process~$Y$, we denote~$Y_{t-}$ denotes the random variable of the left limits,~$Y_{t+}$ that of the right limits, and with $Y^-=(Y_{t-})_{t \geq 0}$ and $Y^+=(Y_{t+})_{t \geq 0}$ the corresponding processes. 
If~$Y$ is a L\'evy process, stochastic continuity implies that 
$Y=Y_-=Y^+$ up to a modification\footnote{A process~$Y$ is said to be a \emph{modification} of a process~$X$ if $\forall t$, $\mathbb P(X_t=Y_t)=1$.}.

The first hitting time of $[t, \infty)$ of~$L$ is the random variable
\begin{equation}\label{Inversetimechange}
H_t := \inf \left  \{ s>0 \, | \: L_s>t \right \},
\end{equation}
which is $\mathcal F$-adapted by the Debut Theorem~\citep{DellacherieMeyer} and has continuous paths if and only if~$L$ is strictly increasing.  The process $H$ is called the \emph{inverse-subordinator} of~$L$. 
Of particular interest for us here is the case where~$L$ is an $\alpha$-stable subordinator, 
i.e. $\psi_L(s)=s^\alpha$, $\alpha \in (0,1)$,
 whose associated inverse-subordinator is central in fractional calculus and anomalous diffusions theory.

Following \cite[Chapter 10]{Jacod}, a \emph{time change} is a non-decreasing, almost surely finite process~$(T_t)_{t\geq 0}$ 
diverging almost surely to infinity for large times.
In particular, both~$L$ and~$H$ are time changes. 
If~$X$ is an $\mathcal F_t$-adapted semimartingale, 
then its time change by~$T$ is the $\mathcal F_{T_t}$-adapted 
semimartingale~$X_T:=(X_{T_t})_{t\geq 0}$.
Further, if~$X$ is almost surely constant on all sets $[T_{t-}, T_t]$ we say that~$X$ is continuous with respect to~$T$; 
in this case many other properties are preserved, and the semimartingale characteristics of~$X$ scale with~$T$.

A \emph{triangular array of random} variables is a collection of random variables $(Y^c_i, J_i^c)_{i \in \NN, c > 0}$ indexed by a scale parameter~$c$ such that 
each $(Y^c_i)_{i\in\NN}$ and $(J^c_i)_{i\in\NN}$ is an i.i.d. sequence,
but not necessarily independent from each other.
For fixed $c$ the variable~$Y^c_i$ retains the interpretation of the $i$-th log-return, and $J_i^c$ the time elapsed between two consecutive price moves.
We can canonically associate to $(Y^c_i, J_i^c)$ two families of continuous-time random walks (CTRWs):
\begin{equation}\label{RW}
R^c_t := \sum_{i=0}^{[t]} Y^c_i
\qquad\text{and}\qquad
T_t^c := \sum_{i=0}^{[t]} J_i^c,
\end{equation} 
and associate to $T^c$ the counting process
$N^c_t := \max\{ n: T^c_n \leq   t \}$.
The notation~$\widehat{\cdot}$ indicates the Fourier transform of probability measures, and the Laplace transform in the time variable is denoted by $\Ll(\cdot, s)$, where~$s$ is the new transformed variable.

\section{The microstructural returns and their analytical properties}\label{microstructure}
At a microscopic level, we postulate that the time series of returns and trade times,
at the time scale~$c$, are determined by a triangular array of random variables $(Y^c_i, J_i^c)_{i\in\NN}$,
where~ $Y_i^c \in \mathbb R$ determines the size of the returns implied by the equity price variation conditional to  observing a price revision,
and~$J_i^c>0$ dictates the time elapsed between subsequent revisions. 
The renewal process~$N^c$ corresponds then to the total number of price movements at~$t$, 
and the tick-by-tick returns process~$\Sigma^c$ is thus given by time-changing~$R^c$ with~$N^c$:
\begin{equation}\label{SubRW}
\Sigma^c_t := \sum_{i=0}^{N^c_t} Y_i^c.
\end{equation} 
At time~$t$ the price will have moved by a quantity $\sum_i^nY_i^c$ if the $n$-th arrival time is recorded before $t$. Or, conditional to $n$ price moves occurred  by time $s$, the price will move again by~$Y_i^c$  before time $t>s$ if the waiting time variable $J^c_{n+1}$ realizes at a value lesser than $t-s$.
We assume that there exists a constant risk-free market rate $r>0$ affecting the price growth linearly in time and independently of the time scale and modify~\eqref{SubRW} as   
\begin{equation}\label{SubRWMod}
\Sigma^{c, *}_t := r t + \sum_{i=0}^{N^c_t} Y_i^c.
\end{equation}

The reasons for this modification shall be explained further on. For the moment, we remark that this physical tick-by-tick model must be understood in the sense that only the price innovations correspond to market observations. Hence, the linear drift introduced in the random walk~$\Sigma^{c, *}$ between two price movements does not give rise to a traded value, 
and impacts the price only at revisions time.
However, further deterministic trends in the price dynamics, such as risk premia, 
are still possible and can be captured by an appropriate choice of~$Y^c$.

\subsection{Joint limits of CTRWs}

The continuous-time pricing model we describe here is based on a scaling limit 
of the CTRW~$\Sigma^{c, *}$ for an appropriately selected triangular array $(Y^c_i, J_i^c)$.
This setup encompasses classical mathematical finance models:
when $(Y^c_i)_{i\in\NN}$ are centered with finite variance and $J^c_i=1$ for all~$i$, 
then the Central Limit Theorem yields a Brownian motion. 
If the~$Y_i^c$ have infinite variance and are in the domain of attraction of a stable process~$X$, 
then their scaling limit yields exactly~$X$. 
 Considering random waiting times for $J^i_{c}$ with finite expectation does not improve here the generality of the setting since by the Renewal Theorem $N_t^c \sim t / \E[J^c_1]$ in probability for large~$t$.  
Therefore, in order to build processes in which the trade time duration information has impact on the distribution of the scaling limit of $\Sigma^c$, one has to consider infinite-mean waiting times.  
Under this choice, taking the limit leads to an anomalous diffusion model for the asset price dynamics. The following result is central to the entire anomalous diffusions theory.

\begin{thm}[Becker-Kern,  Henry, Jurlewicz, Kern, Meerschaert, Scheffler, Straka]\label{CTRWconvergence}
Assume that $(Y^c_i, J_i^c)_{i\in\NN, c>0}$ forms a triangular array of random variables and set $R^c$,  $T^c$ 
and $\Sigma^c$  as in~\eqref{RW}-\eqref{SubRW}. 
If there exists a bivariate L\'evy process $(X,L)$, where~$L$ is a subordinator with inverse process~$H$ 
as in~\eqref{Inversetimechange}, such that
\begin{equation}\label{basicconv}
\lim_{c\uparrow \infty} \left(R^c_{c}, T_{c}^c\right) = \left(X, L\right),
\end{equation} 
in the $J_1$-topology on the Skorokhod space $\D(\R\times \R_+)$, then
\begin{equation}\label{conv}
\lim_{c\uparrow \infty} \Sigma^c =((X^-)_H)^+,
\end{equation}
in the $J_1$-topology on $\D(\R)$, where $((X^-)_H)^+$ is the  right-continuous version of the process obtained by time-changing by $H$ the left limits process of~$X$.
\end{thm}
This theorem has appeared in various forms and has an interesting evolution. 
It was first proved in~\cite{CTRWCoupledLlimits}  under the weaker M1 topology, under an assumption only slightly weaker than independence between spatial evolution and waiting times. However, even if it was the process $X_H$ that was claimed there to be the limit,  the latter can be shown to coincide with $(X^-)_H$ under such assumptions. 
This was remarked by~\cite{StrakaHenry}, who also gave a version of the theorem which allows dependence between~$X$ and~$L$, but excludes the possibility of either~$X$ or~$L$ being a compound Poisson process (CPP). 
Another proof is obtained by combining~\cite[Theorem 3.1 and Remark 3.5]{MeerscahertOCTRW}, which finally extends the result of Straka and Henry to CPPs.

\begin{oss} \label{thmID}
Unless the $J_i^c$  are constant or exponentially distributed, the CTRW limit is not Markovian.
\end{oss}

\begin{esmp} For a sequence $(Y_i)_{i\in\NN}$ of i.i.d. centered random variables with unit variance, 
let $Y_i^c := c^{-1/2} Y_i$;
consider further the i.i.d. sequence $(J_i^c)_{i\in\NN}$ distributed as $\mathrm{Exp}(\lambda)$, 
for some $\lambda >0$. 
As previously detailed applying the Central Limit Theorem and the Renewal Theorem  show the familiar convergence of~$\Sigma^c$ to~$W_{\lambda}$ for some Brownian motion~$W$.
\end{esmp}

\begin{esmp}\label{stableEX} 
Assume that $(Y_i)_{i\in\NN}$ and $(J_i)_{i\in\NN}$ are independent sequences of iid random variables belonging to the domain of attraction of respectively an $\alpha$-stable law~$X$ with $\alpha \in (1,2)$, 
and a $\beta$-stable law~$L$ with $\beta \in (0,1)$, 
namely there exist regularly varying sequences $(B_n)_{n\in\NN}$ and $(b_n)_{n\in\NN}$, 
with respective indices $-1/\alpha$ and $-1/\beta$ such that
$B_n \sum_{i=1}^n Y_i$
and
$b_n \sum_{i=1}^n J_i$
converge respectively to~$X$ and~$L$ almost surely. 
Then letting $Y_i^c := B(c) Y_i$ and $J_i^c:=b(c)J_i$, 
with $B(c) := B_{[c]} $
and $b(c):=b_{[c]}$  yields an explicit triangular array, and the theorem to the stable processes canonically associated 
with~$X$ and~$L$. 
In this case, Theorem~\eqref{CTRWconvergence} collapses to~\cite[Theorem 4.2]{CTRWS}.  

\end{esmp}

\begin{esmp} An explicit representation of the CGMY process as a CTRW limit can been obtained by appropriately tempering variables in the domain of attraction of a stable law, as explained  in~\cite{CGMYCTRW}. Combining this with Example~\ref{stableEX} provides another explicit CTRW limit representation of~\eqref{conv} for a CGMY process~$X$ and a stable subordinator~$L$.  
\end{esmp}

\subsection{Transform analysis and connections to fractional calculus}

It is remarkable that the CTRW limit in Theorem~\ref{CTRWconvergence} enjoys a very high degree of analytical tractability.  
For example the probability density of an inverse L\'evy subordinator~$H$ is known in terms of the L\'evy measure of the original process~$L$. 
Similarly, the law of $X_{H_t-}$ can be recovered by integral transforms involving~$\nu_{X,L}$ and the other Fourier-Laplace characteristics, as explained in~\cite{triangular} and~\cite{MeerscahertOCTRW}.
We recall the following from~\cite[Proposition 4.2]{MeerscahertOCTRW}:
\begin{prop}\label{CTRWtransforms} 
Let $X_{H_t-}$ be as in the CTRW limit in~\eqref{conv}, with law~$P_t$.
Then
\begin{equation}\label{FLT}
\Ll\left(\widehat{P}_{t}(dz),s\right)=\frac{1}{s}\frac{\phi_L(s)}{\psi_{X,L}(z,s)}.
\end{equation}
\end{prop}
The formula of the Laplace transform of $X_{H_t-}$ is particularly simple. 
Having at hand a specification for $X_{H_t-}$ in terms of the involved characteristic exponents, 
by virtue of~\eqref{FLT} we are only one Laplace inversion away from the characteristic function, 
and we shall see that this inversion can be computed explicitly in our cases. 
From a theoretical perspective, the Fourier-Laplace transform of the process provides an interesting connection between the stochastic representation of anomalous diffusions via CTRW limits and the classical characterization of their laws as weak solutions of fractional abstract Cauchy problems. 
For details we refer the reader to~\cite{fractionalST, FKAnomalousTS, MeerscahertOCTRW, triangular}, 
and references therein.

\section{The asset price models}\label{model}

We introduce here the two anomalous diffusions to establish the connection between trades duration and the implied volatility surface.

\begin{defn}\label{modeldefn} 
Let~$X$ be a L\'evy process,~$L$ an independent $\beta$-stable subordinator, 
and $(Y_i^c, J_i^c)_{i\in\NN,c>0}$ a triangular array satisfying~\eqref{basicconv}.
We define the underlying price~$S$ as
\begin{equation}\label{asset}
S_t=S_0\exp( r t +Y_t), \qquad S_0>0,
\end{equation}
with~$Y_t:=X_{H_t-}$ given by~\eqref{conv}, and shall consider the following two cases:
\begin{enumerate}
\item[\textbf{(SL)}] The  \emph{purely subdiffusive L\'evy model} is such that $(Y_i^c, J_i^c)_{i\in\NN}$ satisfy the assumptions of Theorem~\ref{CTRWconvergence}  with $(X,L)$ in the right-hand side of~\eqref{basicconv}; 
\item [\textbf{(DRD)}] The \emph{model with dependent returns and trade durations} is such that 
$(Y_i^c, J_i^c)_{i\in\NN}$ satisfy the assumptions of Theorem~\ref{CTRWconvergence}  with $(X_L,L)$ in the right-hand side of~\eqref{basicconv}.
\end{enumerate}
\end{defn}

The two models look very similar, the only difference being that the second requires convergence of the return innovations to the subordinated L\'evy process~$X_L$ instead of~$X$. 
Yet, this difference is critical since this subordination is precisely what introduces coupling in the DRD model. 
We shall denote the CTRW limits~$Y^{SL}$ and~$Y^{DRD}$ and, correspondingly, 
the price processes~$S^{SL}$ and~$S^{DRD}$.
The underlying standard L\'evy model  is $S^0=(S^0_t)_{t \geq 0}$ with $S_t^0=S_0\exp( r t +X_t)$.

\begin{oss}\label{CTRWirrelevance} 
For the SL model, since~$X$ is stochastically continuous and independent of~$H$,
then $X_{H_t-}=X_{H_t}$  in law for each $t>0$. 
\end{oss}

\begin{oss}\label{levyrevert} 
As $\beta$ tends to~$1$,~$L_t$ tends to $t$ in probability and almost surely. 
Therefore the usual conditional independence argument shows that~$S_t$ tends in law to~$S_t^0$. 
So in the limiting case, the L\'evy models are recovered,
and~$\beta$ can be interpreted as a parameter regulating the divergence from L\'evy, and therefore quantifies the degree of `anomaly'  of the diffusion.  
\end{oss}

\begin{esmp}\label{SubBS} 
When~$X$ is a Brownian motion and~$L$ an independent $\beta$-stable subordinator the resulting SL model is the \emph{subdiffusive Black-Scholes} first introduced in~\cite{SubBS}.
\end{esmp}

\begin{esmp}\label{FPP} 
\cite{CarteaMB} introduce a CTRW model with independent trade duration and returns, 
where the conditional waiting time is modelled through a hazard function. 
They in particular consider the latter to be of Mittag-Leffler type $\P(T_n>t)=E_\beta(-t^\beta)$ 
(see also~\eqref{ML} below), 
and the price innovations follow an arbitrary infinitely divisible distribution. 
The resulting driving CTRW is a Fractional Poisson process (FPP) as in~\cite{Laskin, vietnam} 
with parameter~$\beta$. 
Since an FPP can be represented as a CPP, time-changed by an independent inverse $\beta$-stable subordinator
(as proved by~\cite{FracIS}), the FPP model by~\cite{CarteaMB} is included in our framework. 
\end{esmp}

\begin{esmp} 
The original model in~\cite{CTRWfinance} and~\cite{CTRWfinance2} also admits an FPP representation, where the returns innovations follow a stable distribution, 
and can be written in terms of a triangular array limit (\citealt{CTRWfinanceCoupled}). 
\end{esmp}
\begin{esmp} A comprehensive treatment of subdiffusive asset models obtained as fractional counterparts of popular L\'evy models is provided in~\cite{CarteDCNegrete}, who tackle the option pricing problem by numerically solving the fractional partial differential equations characterizing their transition probabilities. In view of the results of \citep{triangular}, all such models admit  a time-changed representation of SL type.
\end{esmp}

We recall that a stable subordinator has no drift: therefore the sample paths of~$Y^{SL}$ and~$Y^{DRD}$ are Lebesgue almost everywhere constant~\cite[Chapter 2]{BertoinSub}, and thus conveniently capture the idea of tick-by-tick trading and persistence of trade duration at all time scales. 
This also implicates that all equivalent measures for~$Y$ are mutually singular with respect to the usual diffusion processes. However, the discounted asset value necessarily contains a Lebesgue absolutely continuous part, orthogonal to all equivalent martingale measures for~$Y$, coming from discounting by the market numeraire 
(the bank account). Therefore, in order for the Fundamental Theorem of Option Pricing~\citealt{FTAP94} to apply, 
we need to cancel such part. This clarifies the choice~\eqref{SubRWMod} of modelling the interest rate effects externally to~$Y$.  
Of course, nothing prevents that the physical dynamics~$Y$ itself have a drift in the component~$X$.  
In Figure~\ref{plotsPaths} we show sample paths of~$H$ and~$Y^{SL}$ when~$X$ is a standard Brownian motion, for two different values of~$\beta$. 
As~$\beta$ increases, reversion to respectively the linear time and a standard Brownian return model with no trades duration effects is observed.

The non-Markovian structure of the two processes captures the possible memory effects in price formation when observing random waiting times between trades. 
As we shall see later, both the value of the process at time~$t$ 
and the time elapsed since the last price revision influence the price evolution.
Dependence between trade times and price returns is a widely acknowledged fact, as pointed out in~\cite{ACD} and confirmed in several empirical studies. This makes the DRD model more realistic compared to the SL one, although the cost/benefit impact in terms of performance of embedding this feature remains to be assessed. For now, observe that  the two models have the same number of parameters, 
so that modelling price/duration dependence does not add any dimension in the calibration and estimation.

\smallskip

It would be useful to find a DRD model  representation in terms of an independent time change similar to the one for the SL model. 
Consider first  the special case $X=L$ in Theorem~\ref{CTRWconvergence}.
Since~$L^-$ is $\mathcal F_t$ adapted then
\begin{equation}
L^H:=((L^-)_H)^+
\end{equation} is an $\mathcal F_{H_t}$-adapted time change.  We can then consider the process~$X$ time changed by~$L^H$ and see in which relation are the process $X_{L^H}$ and the limiting process $((X_L)^-_{H})^+$ coming from Theorem~\ref{CTRWconvergence}, where we recall that $X_L$ is just a L\'evy subordinated process. Although these processes bear some resemblance their construction is different: in particular now~$X$ is independent of the time change. However, 


\begin{equation}
(X_L^-)_H=X_{(L^-)_H}
\end{equation}
because~$X$ has left limits and~$L$ is increasing. But then
\begin{equation}
((X_L^-)_H)^+=(X_{(L)^-_H})^+=X_{((L^-)_H)^+}=X_{L^H}
\end{equation}
since~$X$ is right-continuous.
There are then two ways of looking at the DRD returns process. The CTRW limit definition gives us a \emph{dependent} representation using a \emph{continuous} time change. The equalities above 
 give instead an \emph{independent} representation employing a \emph{discontinuous} time change. 
Both will be useful in the sequel.

Finally, let us briefly comment on the nature of the process~$L^H$. It is easy to show that, 
for any $t\geq 0$, 
\begin{equation}
L_{H_t-}=\sup\{s<t:  s=L_u, \text{ for some }u  \geq 0  \}.
\end{equation}
 In light of this identification, the process $(L_{H_t-})_{t \geq 0}$ is sometimes called the \emph{last passage process} \cite[Chapter 1]{BertoinSub} and plays an important role in potential theory for L\'evy processes: it can be seen as the discontinuous increasing process which represents the past time at which $H$ has started resting at its current (time-$t$) location. 
Now the times of discontinuity of $(L_{H_t-})_{t \geq 0}$ coincide with the points in the image of~$L$ isolated on their right, which on account of~$L$ being driftless, by~\cite[Chapter 1, Proposition 1.9]{BertoinSub} is a set of Lebesgue measure zero. From this we conclude that~$L^H$ is a right-continuous modification of the first passage time.
Moreover the post-jump value of~$L^H$ is exactly~$t$, and in any case~$L^H_t \leq t$ almost surely. This ties in with the interpretation of~$L^H$ as a delayed calendar time. 

 \smallskip

Notably, the distribution of~$L_t^H$ is known. We have the following result:

\begin{prop}\label{betaLaw} 
 Denote by $\Bb_{a, b}$ the Beta distribution with parameters~$a$ and~$b$. For any $t > 0$, $L^H_t$ is distributed as $t\Bb_{\beta, 1-\beta}$.
\end{prop}

\begin{proof}
See~\cite[Example 5.5]{CTRWCoupledLlimits} or~\cite[Example 5.2]{MeerscahertOCTRW}.
\end{proof}

This underpins the greater analytic tractability of the DRD model with respect to the SL model:  somewhat paradoxically, the more realistic model is also the more explicit.
Proposition~\ref{betaLaw} clarifies how the DRD model captures the paradigm 
of~\cite{Engle} and~\cite{EngleDufour}. The DRD time-changed evolution obeys a form of delayed calendar time whose mass in $[0, t]$ concentrates more around 0 or $t$ depending on whether~$\beta$ is close to zero or one (Figure~\ref{betas}). 
This mass represents the quantity of delay one has to apply to~$X$ to obtain the current price value. 
When~$L$  has a low $\beta$, that is when duration of trade is higher, 
the price evolution is stickier, since~$t \Bb_{\beta, 1-\beta} $ is much smaller than~$t$ 
with high probability. 
This is associated with a reduced impact of the individual trades on the price process because the informational content of sporadic trading is low. 
Conversely, as $t \Bb_{\beta, 1-\beta} $ is close to~$t$ with high probability 
(namely when~$\beta$ is close to one) we observe a higher trading activity, typically associated with the presence of informed traders.  
In such a case the contribution of each single trade to the process of price formation is greater, and the impact of trading on price higher. A similar reasoning applies to the SL model. Here combining subordination with independence `delays' the evolution of~$X$ for the time necessary to the next price revision to happen, but the resulting move retains the variance of an earlier point-in-time position of the process~$X$. 
Therefore, again, the lower the~$\beta$, the stickier the price dynamics.

\section{Moments and time series properties}\label{statistical}
 
We derive some statistical properties of the SL and DRD models and provide some initial insight 
on the structure of the volatility surface they generate, anticipating the full analysis in Section~\ref{surface}.
We begin with the moments of the DRD model, whose analytic tractability plays a major role. 
The following proposition extends~\cite[Theorem 2.1]{LeonenkoCorr} to higher cumulants.
In this section,~$X$ is a given L\'evy process, $T$ an independent time change, 
and we let~$\kappa_i$ and~$\tau_i$ denote their respective $i$-th cumulants,
which we assume to exist for $i=1, \ldots, 4$.

\begin{prop}\label{moments}
The process~$Y:=X_{T}$ has moments up to order four, and its cumulants read
\begin{equation}\label{momentsTC}
\begin{array}{lllcll}
& \kappa^Y_1 &=  \tau_1 \kappa_1, 
& & \kappa^Y_2 &=   \tau_1 \kappa_2 + \kappa_1^2 \tau_2,\\
& \kappa^Y_3 &=    \tau_1 \kappa_3 + 3 \kappa_1 \kappa_2 \tau_2  + \kappa_1^3 \tau_3,
& & \kappa^Y_4 &=    (3 \kappa_2^2  + 4 \kappa_1 \kappa_3) \tau_2 + 6 \kappa_1^2 \kappa_2 \tau_3 + \kappa_4 \tau_1 + \kappa_1^4 \tau_4. 
\end{array}
\end{equation}
\end{prop} 

\begin{proof}
In our notation
$\kappa_n=-  \left(\ii^n \psi^{(n)}_X(0) \right)$.
We proceed as in~\cite[Theorem 2.1]{LeonenkoCorr}, where the usual conditioning argument yields
\begin{equation}\label{k1t1}
\E[Y_t]=  \ii \frac{d}{d z} \E\left[e^{- \ii z Y_t}\right] \Big|_{z=0}
= \ii \frac{d}{d z} \E\left[e^{- \psi_X(z) T_t}\right]\Big|_{z=0}=- \ii \psi'_X(0) \E[T_t],
\end{equation}
which gives~$\kappa^Y_1$. 
Next
\begin{equation}\label{var}
\E[Y^2_t]= - \frac{d^2}{d^2 z} \E\left[e^{- \ii z Y_t}\right] \Big|_{z=0}
= \psi''_X(0)\E[T_t ] - \psi'_X(0)^{2} \E\left[T^2_t\right],
\end{equation}
Subtracting from (\ref{var}) the square of (\ref{k1t1}) reconstructs $\tau_2$ and yields  $\kappa^Y_2$.
Similarly,
\begin{equation} 
\E[Y^3_t]= - \ii \frac{d^3}{d^3 z} \E\left[e^{- \ii z Y_t}\right] \Big|_{z=0}
 = -\psi'''_X(0)\E[T_t ] + 3 \E\left[T^2_t\right] \psi'_X(0)\psi''_X(0) + \ii \psi'_X(0)^3\E\left[T^3_t\right];
\end{equation}
 calculating $\E[Y^3_t]- 3 \E[Y_t] \E[Y^2_t]+ 2 \E[Y_t]^3$ and factoring the $\tau_i$ as necessary we obtain~$\kappa^Y_3$.
The last term~$\kappa^Y_4$ is obtained analogously.
\end{proof}

The above proposition confirms the well-known fact that a L\'evy model~$X$ subordinated by a L\'evy process~$L$  creates non-zero skewness and kurtosis even in the presence of a mesokurtic and symmetric parent process~$X$ such as a Brownian motion. 
Our situation here is identical, and carries the message that trade duration alone can be a determinant of departure from normality of returns (thus, in an option pricing perspective, creating volatility smile). 
However, the term structure analysis of the moments is completely different. 
The key fact is that  the moment time dispersion of a time-changed L\'evy process only depends on the moments of the time change, and not on the moments of~$X$. 
In the usual L\'evy subordination case, that is when~$T$ is a L\'evy process, one then sees that the moments are linear in~$t$, consistently with the fact that the subordinated process is itself L\'evy. 
As a consequence the skewness and kurtosis of the returns vanish with time. 
In contrast, our framework produces a nonlinear time evolution of the moments, 
which we analyse in detail for the DRD model, where such evolution is polynomial.
 
\begin{prop}\label{BTLmoments} 
For any $t\geq 0$, the first four cumulants of~$Y_t^{DRD}$ are
\begin{equation}\label{cumulantsBTL}
\begin{array}{rl}
\kappa^Y_1 & = \displaystyle \beta   \kappa_1 t,\\
\kappa^Y_2 & = \displaystyle \beta \kappa_2 t +\frac{\kappa_1^2}{2} (1- \beta)\beta t^2,\\
\kappa^Y_3 & = \displaystyle \beta \kappa_3 t + \frac{3\kappa_1 \kappa_2}{2}(1-\beta) \beta t^2
 - \frac{\kappa_1^3}{3}(1-\beta)\beta(2 \beta-1) t^3,\\
\kappa^Y_4 & = \displaystyle \beta \kappa_4 t +
\frac{4\kappa_1\kappa_3+3\kappa_2^2}{2}\beta(1-\beta) t^2\\
 & \quad -2 (1- \beta)\beta (2 \beta-1) \kappa_1^2 \kappa_2 t^3
  + \frac{\kappa_1^4}{8}(1-\beta)\beta\left(2-11\beta(1-\beta)\right) t^4,
\end{array}
\end{equation}
and the following asymptotic relations hold:
\begin{equation}
\begin{array}{lllllll}
& \displaystyle \lim_{t \uparrow \infty}\Skew(Y_t) & = & \displaystyle \frac{2 \sqrt 2}{3} \frac{1-2 \beta}{\sqrt{ (1-\beta)\beta  } } \mathrm{sgn}(\kappa_1),
& \displaystyle \lim_{t \uparrow \infty}\Kurt(Y_t) & =  & \displaystyle \frac{1}{\beta(1-\beta)}  -\frac{11}{2},\\
& \displaystyle \lim_{t\downarrow 0}\sqrt{t}\ \Skew(Y_t) & = & \displaystyle \frac{\kappa_3}{\sqrt{ \beta \kappa_2^3} },
& \displaystyle \lim_{t\downarrow 0} t\ \Kurt(Y_t) & = & \displaystyle \frac{\kappa_4}{\beta \kappa_2^2}.
\end{array}
\end{equation}
\end{prop}

\begin{proof}
By explicitly integrating the Beta probability density function we have the central moments of $T_t$:
\begin{align}
\mu_1^T &=   \E\left[L_t^H\right] =   \beta  t  =\tau_1, \label{cumulantsBeta1} \\ 
\mu_2^T &=   \VV\left[L_t^H\right] =  \frac{1}{2} (1- \beta)\beta t^2=\tau_2,	\label{cumulantsBeta2}   \\
\mu_3^T &= \E\left[(L_t^H - \tau_1)^3\right]  = - \frac{1}{3}(1-\beta)\beta(2 \beta-1) t^3= \tau_3, \label{cumulantsBeta3} \\
\mu_4^T &=  \E\left[(L_t^H - \tau_1)^4\right]
 = \frac{\beta}{8}(1-\beta)\left(2 -11 (1-\beta) \beta)\right)  t^4=\tau_4 + 3 \tau^2_2.  \label{cumulantsBeta4}
\end{align}
Since in the DRD model $X_t$ and~$L_t^H$ are independent, we can solve the above equations for $\tau_i$ and substitute in~\eqref{momentsTC}
obtaining~\eqref{cumulantsBTL}. 
Calculating further the normalized cumulants $\Skew(Y_t)=\kappa^Y_3/(\kappa^Y_2)^{3/2}$
and $\Kurt(Y_t)=\kappa^Y_4/(\kappa^Y_2)^2$ and taking respectively the limits for large~$t$
and the leading order around $t=0$ imply the limits in the proposition.
\end{proof}


In the DRD model, as the time scale gets larger, higher moments do not vanish, 
but converge to a level that only depends on~$\beta$, and not on the value of the L\'evy cumulants 
(the sign of~$\kappa_1$ dictates the sign of the skewness). 
As frequently noted, leptokurtosis and negative skewness of returns are important drivers of implied volatility smiles. 
It thus makes sense to deduce that non-zero time limits of skewness and excess kurtosis determine persistence of the volatility smile over time. 
In contrast, for~$t$ close to zero, moment explosions are observed, as in the L\'evy case;
the rate of this explosion is exactly that of exponential L\'evy models, including--up to a normalization by $\beta$--the constant factor. 
This suggests that the short-term smile/skew behaviour of the DRD implied volatility 
should be identical to that of the underlying L\'evy model.
We will verify these intuitions and make the matters more precise in Section~\ref{surface}.

The analysis of the returns series properties stems from the observation that the models we are studying, although not Markovian with respect to their own filtration, admit a Markovian embedding. 
Remarkably, the marginal distributions of this embedding are known for the DRD process. 
For any $t\geq 0$, we define the backward renewal time
\begin{equation}
V_t := t- L_t^H,
\end{equation}
which represents the time elapsed from the current instant~$t$ to the previous price move. Knowing the price at~$t$ and the time since the last price move is enough to fully describe the law of the future asset evolution.
\begin{prop}\label{statisticalProp} The following properties hold:
\begin{itemize}
\item[(i)] the pairs $(Y^{SL}, V)$ and $(Y^{DRD}, V)$ are time-homogeneous Markov processes;
\item[(ii)] the process~$Y^{SL}$ has correlated increments, 
whereas~$Y^{DRD}$ has uncorrelated increments; 
\item[(iii)] the increments of~$Y^{SL}$ are non-stationary, whereas the increments of~$Y^{DRD}$ are weakly stationary;
\end{itemize} 
\end{prop}

\begin{proof} 
Item~(i) is proved in~\cite[Theorem 4.1]{MSMarkov}. 
For the SL model, statement~(ii)  can be deduced from~\cite[Example 3.2, Equation 9]{LeonenkoCorr}, since in our case $\E[X_1] \neq 0$. 
 In the case of the DRD model, for $s \leq t$, we can write 
(we drop the model superscript for convenience)
$$
\E[X_t  X_s]
 = \E[(X_t-X_s) X_s]+\E\left[X^2_s\right]
 = (t-s)s \E[X_1]^2+ s \VV[X_1] + s^2 \E[X_1]^2
  = t s \E[X_1]^2+ s \VV[X_1],
$$
so that by independence and conditioning 
\begin{align}\label{covY}
\Cov(Y_t, Y_s)
 & = \E\left[L^H_t L^H_s\right] \E[X_1]^2+ \E\left[L^H_s\right] \VV[X_1]
 - \E\left[L^H_t\right] \E\left[L^H_s\right] \E[X_1]^2 \nonumber \\ 
 & = \Cov\left(L^H_t, L^H_s\right)\E[X_1]^2 + \E\left[L^H_s\right] \VV[X_1].
\end{align}
Thus, considering increments and using the above, together with Proposition~\ref{moments},
\begin{align}
\Cov(Y_t-Y_s, Y_s)
 & = \Cov(Y_t, Y_s)- \VV[Y_s]
 = \E[X_1]^2 \left(\Cov(L^H_t, L^H_s)- \VV[L^H_s]\right) \nonumber \\
 & = \E[X_1]^2 \Cov(L^H_t- L^H_s, L^H_s),
\end{align}
so absence of returns autocorrelation is equivalently checked on~$L_t^H$.  
Now~\cite[Example 5.4]{MSMarkov} give the conditional transition probabilities $p_t(y_0, v_0, dy, dv):=\P(L^H_t \in dy , V_t \in dv \, | \; y_0, v_0 )$ of the Markov process $(Y_t, V_t)$ as:
\begin{align*}
p_t(y_0, 0, dy, dv)&=\frac{v^{-\beta}}{\Gamma(1-\beta)}\frac{(t-v)^{\beta-1}}{\Gamma(\beta)}\delta_{y_0+t-v}(dy)dv \mathds{1}_{\{0 < v <t \}}, \\
p_t(y_0, v_0, dy, dv)&= \delta_{y_0}(dy)\delta_{v_0+t}(dv) \left(\frac{v_0+t}{v_0}\right)^{-\beta} \nonumber \\  \phantom{hhhhhhhhhhhhhh} &+ \left(\int_{v_0}^{v_0+t} \left(\frac{v}{v_0}\right)^{-\beta}\delta_{v_0+y_0+t-v}(dy)\frac{(v_0+t-s-v)^{\beta-1}}{\Gamma(\beta)}\frac{\beta s^{-\beta-1}}{\Gamma(1-\beta)}ds \right) dv.
\end{align*}
Explicitly integrating the second line we have
$$
p_t(y_0, v_0, dy, dv)= \delta_{y_0}(dy)\delta_{v_0+t}(dv) \left[\frac{v_0+t}{v_0}\right]^{-\beta} + \delta_{v_0+y_0+t-v}(dy)  \left[\frac{t-v}{v}\right]^{\beta}   \frac{(t-v+v_0)^{-1}}{\Gamma(\beta) \Gamma(1-\beta)} dv,
$$
whence,  for $t_2> t_1$ 
the the joint probability densities $P_{t_1, t_2}$ for $(L^H_{t_1},V_{t_1}, L^H_{t_2},  V_{t_2} )$ can be obtained through the Chapman-Kolmogorov equation
$$
P_{t_1, t_2}(d y_1,d v_1,d y_2, d v_2 )=p_{t_1}(0, 0, dy_1, dv_1)p_{t_2-t_1}(y_1, v_1, dy_2, dv_2). 
$$
Integrating out $dv_1$ and $dv_2$ from the explicit form of the above for $0<v_1<t_1-y_1$, $0<v_2<t_2-y_2$ leads to the joint density of $(L^H_{t_1},L^H_{t_2})$:
\begin{align}
P_{t_1, t_2}(dy_1, dy_2)=  & \frac{y_1^{\beta-1} [(t_1-y_1)(t_2-y_2)]^{-\beta}  (y_2-t_1)^{\beta}}{[\Gamma(1-\beta) \Gamma(\beta)]^2 (y_2-y_1)}\mathds{1}_{\{0< y_1 < t_1 < y_2 < t_2 \}}  dy_1 dy_2 \, \nonumber \\ & +  \frac{(t_2 -y_1)^{-\beta} y_1^{\beta-1}}{\Gamma(1-\beta) \Gamma(\beta)}\delta_{y_2}(dy_1)dy_2.
\end{align}
Setting $t_1=t$ and $t_2=t+h$, a long integration yields
\begin{align}\label{covL}
&\Cov(L^H_{t+h} , L^H_{t} )=\int_{\mathbb R^+ \times \mathbb R^+} y_1 y_2 P_{t, t+h}( dy_1, dy_2) -\beta^2 t  (t +h) \nonumber \\
=&\int_t^{t+h} \int_0^{t}\frac{y_1^{\beta-1} ((t_1-y_1)(t_2-y_2))^{-\beta} (y_2-t_1)^{\beta} }{(\Gamma(1-\beta) \Gamma(\beta))^2 (y_2-y_1)}  dy_1 dy_2 +\int_0^t  \frac{(t+h -y_1)^{-\beta} y_1^{\beta+1}}{\Gamma(1-\beta) \Gamma(\beta)}dy_1 -\beta^2 t  (t +h) \nonumber \\
 & = \frac{1}{2}t \beta(t+ 2 h \beta + t \beta) -\beta^2 t  (t +h)  =\frac{1}{2}t^2(1-\beta)\beta=\VV[L^H_t],
\end{align}
and therefore $\Cov(L^H_{t+h}- L^H_t, L^H_t)=\Cov(L^H_{t+h} , L^H_t)-\VV[L^H_t]=0$,
which shows that the increments of the DRD model are uncorrelated, and~(ii) holds.

Finally, using \cite[Corollary 3.3]{CTRWS} together with a conditional argument,
we see that the expected value of the increments of~$Y^{SL}$ depends on~$t$, so that these cannot be stationary.  Combining $\E[Y^{DRD}_{t+h}-Y^{DRD}_{t}]=\E[X_1]\beta h$ with absence of correlation between increments in the DRD shows weak stationarity and finishes the proof of~(iii).
\end{proof}

It is generally accepted that returns times series calculated at lags of above a couple of minutes show no autocorrelation. Stationarity is also a desirable statistical property shown by the returns: 
both these stylized facts are captured by the DRD model, which in this respect is strikingly similar to a L\'evy process. 
However, these properties are not featured by the SL model, further suggesting that the DRD model might be preferable. 


\section{Measure changes and derivatives valuation}\label{price}

\subsection{Equivalent martingale measure changes}

In order to apply classical valuation theory, one needs to show that the physical dynamics admit a martingale specification and to identify (if possible) an explicit equivalent martingale measure.  
In our models, there are two sources of market risk: the uncertainty in the returns distribution, and the trade duration, captured respectively by the processes~$X$ and~$L$. We could in principle consider measure changes affecting the dynamics of both these processes. However, $\alpha$-stable processes are not stable by equivalent measure change, since the Hellinger distance of the L\'evy measures of any two stable subordinators is infinite. For  example, if a standard Esscher transform is used, after measure change the process becomes tempered stable. Hence, since we are interested in the risk-neutral parametrizations of the SL and DRD models, we shall restrict our analysis to the class of equivalent martingale measures that only involve transformation of the law of~$X$. 
As one may reasonably guess, such a class coincides with the set of equivalent measures under which the underlying L\'evy model~$S^0$ is itself a martingale.

\begin{prop}\label{EMM}
Let $S$ be of SL or DRD type under $\P$, and~$\QQ\sim\P$ an equivalent measure such that~$(e^{-rt}S^0_t)_{t\geq 0}$ is a L\'evy exponential martingale under~$\QQ$. 
Define the Radon-Nikodym derivative $Z := d \QQ/ d \P$, consider its time change~$Z_H$, 
for any~$H$ of the form~\eqref{Inversetimechange},
and introduce the measure~$\widetilde{\QQ}$ via $d\widetilde{\QQ}/d \P:=Z_{H}$.
Then $\left(e^{-r t} S^{DRD}_t\right)_{t\geq 0}$ and  $\left(e^{-r t}S^{SL}_t\right)_{t\geq 0}$  are martingales respectively under $\QQ$ and $\widetilde{\QQ} $.
\end{prop}

\begin{proof}
For the $S^{SL}$ models the statement follows from \citep[Theorem 2]{TorricelliTLS} by taking~$L_t$ to be a $\beta$-stable subordinator, $S_t=e^{-rt}S^{SL}_t$, $\mathcal X_t=Z_t$ and $\mathcal H_t=1$, which corresponds to no change of measure in the subordinator. 
For the DRD model it suffices to observe that~$L^H$ is a bounded family of stopping times and thus $e^{-r t}S^{DRD}=\exp\left(X^*_{L^H}\right)$ is a martingale 
under~$\QQ$ by Doob's Optimal Sampling Theorem.
\end{proof}

Again we emphasize that this is a subset of all the possible equivalent martingale measures
and that for technical reasons we ignore a market price of duration risk. 
A model in which this risk can be priced can be obtained for example by considering for~$L$ 
the wider class of tempered stable subordinators, which is closed under the Esscher transform. 
This class, along with related questions of market completeness, is studied in~\cite{TorricelliTLS}; 
see also~\cite{TorricelliFries} for the situation when trade duration is caused by market suspensions.

\subsection{The pricing formula}

Having established that the risk-neutral specification comes in the form of a time-changed martingale exponential, Proposition~\ref{FLT} can be combined with standard integral price representations 
to yield semi-closed-form valuation formulae. 
Remarkably, the characteristic functions of the log-price in the SL and DRD models admit a very simple representation in terms of the one-parameter Mittag-Leffler function
\begin{equation}\label{ML}
  E_a(z) := \sum_{k=0} ^\infty \frac{z^k}{\Gamma(a k + 1)},
  \end{equation}
where~$\Gamma$ is the usual Gamma function, and of the confluent Hypergeometric function
\begin{equation}\label{1F1}
\hyp(a,b ; z) := \sum_{k=0} ^\infty \frac{ (a)_k}{(b)_k}  \frac{z^k}{k!}.
 \end{equation}
\begin{thm}\label{pricingThm}
Let~$Y$ be either process in Definition~\ref{modeldefn}, and $F(\cdot)$ a contingent claim on~$S$ maturing at~$T$. 
Assume that $x\mapsto f(x):=F(e^ x)$ is Fourier-integrable  and let $\mathcal S_f$ be the domain of holomorphy of its Fourier transform $\widehat f$. 
Let $\Phi_t(z):=\E[e^{-\ii z Y_t}]$ be the characteristic function of~$Y_t$ taken according to the relevant measure as described in Proposition~\ref{EMM}, denote by~$\mathcal S_Y$ its holomorphy domain, 
and assume $\mathcal S_f \cap \mathcal S_Y \neq \emptyset$. The price $P_0$ of the derivative paying~$F(S_T)$ at time $T$ is given by
\begin{align}\label{pricingLewis}
P_0&  =\E\left[ e^{-rT} F( S_T)\right]
 = \frac{e^{-rT}}{2 \pi} \int_{\ii \gamma - \infty}^{\ii \gamma + \infty} \frac{\Phi_T(z) \widehat f(z)}{\left(S_0e^{rT}\right)^{\ii z} } dz.
\end{align}
The value $\gamma \in \mathbb R$ is chosen such that the integration line lies in $\mathcal S_f \cap \mathcal S_Y$ and
\begin{equation}\label{CFS}
\Phi_t(z) = 
\left\{
\begin{array}{ll}
\displaystyle E_\beta\left(- \psi_X(z) t^\beta\right), & \text{if } Y=\YS,\\
\displaystyle\phantom{,}_1F_1(\beta,1, -t \psi_X(z)), & \text{if } Y=Y^{DRD}.
\end{array}
\right.
\end{equation} 
\end{thm}
\begin{proof}
Under the given assumptions, the Plancherel representation~\eqref{pricingLewis} 
is standard (see~\cite{Lewis} for example), and we only need to prove~\eqref{CFS}. 
In the SL model, by independence of~$X$ and~$L$ we have
$\psi(s, z)=\phi_L(s)+ \psi_X(z)=s^\beta+ \psi_X(z)$, 
and Proposition~\ref{CTRWtransforms} then yields
\begin{equation}
\Ll (\Phi_t(z), s)= \frac{  s^{\beta-1}}{s^{\beta} + \psi_X(z)}.
\end{equation}
Inverting the right hand-side, as in~\citep{HMS}, one obtains~\eqref{CFS}.  
In the DRD  model after conditioning and applying Proposition~\ref{betaLaw}, we obtain
\begin{equation}
\Phi_t(z) = \E\left[\exp\left(-\psi_X(z ) L^H_t\right)\right]
 = \E\left[\exp\left(-t \psi_X(z)\Bb_{\beta, 1-\beta}\right)\right],
\end{equation} 
and the statement follows from the characteristic function of $\Bb_{\beta, 1-\beta}$.
\end{proof}

\begin{oss}
Fast computational routines for the Mittag-Leffler and the confluent hypergeometric functions are available in most software packages.  
Also, the two functions can be unified in a single software implementation by observing that the three-parameter Mittag-Leffler function
  \begin{equation}\label{ML2}
  E_{a,b,c}(z)=\sum_{k=0} ^\infty (c)_k\frac{z^k}{\Gamma(ak +  b )}
  \end{equation}
is such that $ E_{a,1,1}(z)= E_a(z)$ and $E_{1,1, c}(z)=\hyp(c,1,z)$. 
Furthermore if $a=b=c=1$, then~\eqref{ML2} reverts to the standard exponential,
which is consistent with the fact that~$S^{SL}$ and~$S^{DRD}$ 
revert to the exponential L\'evy model~$S^0$.
\end{oss}

\begin{oss}
The function~$E_\beta$ is entire and  $\phantom{.}_1F_1(\beta,1, -t \psi_X(\cdot))$ 
is regular in the complex plane without the negative real axis;
hence $\mathcal S_f \cap \mathcal S_Y \neq \emptyset$ 
depends on the 
domain of~$\psi_X$ and~$\widehat f$ only.
\end{oss}

\begin{oss} 
If~$X_{H}$ has an FPP  structure, 
then~\eqref{pricingLewis} coincides with the formula 
given by~\cite[Theorem 3]{CarteaMB}, 
when the jump sizes have infinitely divisible distribution.
\end{oss}

One sees that the pricing formulae are formally obtained from the standard L\'evy case by replacing the exponential function with two different kinds of `stretched exponentials'. 
The parameter~$\beta$  relaxes the shape of the characteristic function, in particular in the tails, 
thereby generating large-maturity prices very different from the base case. 
This overcomes the `curse of exponentiality' of the standard models (both L\'evy and exponentially-affine), 
for which the long-maturity option prices follow Laplace-type asymptotics of leading order $\exp(-T)/\sqrt{T}$. 
We will detail this better, together with its implications on the volatility surface, 
in Section~\ref{surface} below. 
Note that the two functions~\eqref{ML} and~\eqref{1F1} have very different behaviours.
In Figure~\ref{functionals}, we can see for example that~\eqref{ML} has a cross-over region where its decay transitions from super to sub-exponential, whereas
in~\eqref{1F1}, the integrand always dominates the exponential. 
This has a clear impact on the shape of the volatility surface, 
as illustrated numerically in Section~\ref{numerical}.


\section{Time asymptotics of the volatility surface}\label{surface}

Bearing in mind the discussion so far, we naturally expect implications of trade duration 
(at least in the form we chose to model it) on the  volatility surface.
The anomalous diffusions processes we constructed are subdiffusions, and as such have a slower distributional dispersion  rate 
 than the benchmark L\'evy models, hence a slower option price convergence for large maturity. 
That said, since Black-Scholes is a L\'evy model, inversion of the Black-Scholes formula using subdiffusive option prices should generate a vanishing implied volatility term structure in order to match the slower price time evolution. 

Less intuitive is to find a reason why the long-term skew should decline slower than standard L\'evy and stochastic volatility models. A first answer is provided by Section~\ref{statistical}: skewness and kurtosis in our models do not tend to zero as time grows but converge to some strictly positive level. 
Therefore Gaussian temporal returns aggregation 
is precluded, and time reversion to a flat volatility might be pushed further away in time\footnote{Gaussian aggregation is by no means responsible of the smile flattening, as shown by~\cite{RogersFlattening}.}. However, as we shall show, an exhaustive answer is provided by the fact that skew and level of the implied volatility are connected, and the property of a vanishing asymptotic implied volatility is sufficient to hamper the skew time decay.

In this section we generically indicate with~$\QQ$ any of the two risk-neutral measures of
Proposition~\ref{EMM}. Without loss of generality, we assume here $r=0$ and $S_0=1$ and, 
denote~$C(K,T)$ the Call option price with strike~$K$ and maturity~$T$.
In the Black-Scholes model $d S_t = \sigma S_t d W_t$, with $\sigma>0$, 
the price of such a Call option is given by
\begin{equation}
C_{\BS}(K,T,\sigma) = S_0\, \Nn\left(d\left(\sigma\sqrt{T}\right)\right) - K\Nn\left(d\left(\sigma\sqrt{T}\right)-\sigma\sqrt{T}\right),
\end{equation}
where $d(z):=-\frac{\log(K)}{z} + \frac{z}{2}$,~$\Nn$ denotes the standard Gaussian cumulative distribution function,
and~$n$ its derivative, the Gaussian density function.
For $K, T \geq 0$, the implied volatility $\sigma(K,T)$ 
is the unique non-negative 
solution to 
$C(K,T) = C_{\BS}(K, T, \sigma(K,T))$, 
and the implied volatility skew is defined as
\begin{equation}\label{skew}
\S(K, T) := \frac{\partial \sigma}{\partial K}(K,T).
\end{equation}
It is known by~\cite{RogersFlattening} that~$\S(K,\cdot)$ converges to zero as the maturity increases, for each~$K$. 
We begin with the following model-free lemma which, under some mild assumptions on the underlying distribution, connects the time decay of the skew with its level.

\begin{lem}\label{asymptotics} 
Let $(S_t)_{t\geq 0}$ be a martingale such that the law
of $S_t$ is absolutely continuous for each~$t$ and converges to zero in distribution 
as~$t$ tends to infinity.  
\begin{itemize}
\item[(i)] For any $K\geq 0$, if $\lim\limits_{T\uparrow\infty}\sqrt{T}\sigma(K,T) =\infty$ 
then, as~$T$ tends to infinity,
\begin{equation}\label{longSkew}
\S(K, T) = \frac{2}{T  \sigma(K,T )} \left[1 + \frac{2\log(K)-4}{T\sigma(K,T)^2} 
 + \Oo\left(\frac{T^{-2}}{\sigma(K,T)^{4}}\right)  \right] - \frac{\QQ(S_T \geq K)}{\sqrt{T}n(d(\sigma(K,T)\sqrt{T}) )};
\end{equation}
\item[(ii)] as $T$ tends to zero, 
\begin{equation}\label{shortSkew}
\S(1, T) = \sqrt{ \frac{2 \pi}{T}} \left[\frac{1}{2}- \QQ(S_t \geq 1)- \frac{\sigma(1,T) \sqrt{T}}{ 2 \sqrt{2 \pi}} +  \Oo\left(\sigma^2(1,T) T\right) \right].
\end{equation}
\end{itemize}
\end{lem}

\begin{proof} We only prove the first statement, as the second one is proved in~\cite[Lemma 2]{gerholdST}.
Since $S_t$ has an absolutely continuous law, then by~\cite[Lemma C.1]{LevyLargeSkew},
$\mathcal S$ in~\eqref{skew} exists, $\partial_K C(K,T)=-\QQ(S_T \geq K)$, and the chain rule yields
\begin{equation}
\S(K, T)=-\frac{\partial_K C_{\BS}(K, T, \sigma(K,T))  +\QQ(S_T \geq K)}{\partial_\sigma C_{\BS}(K, T, \sigma(K,T))}.
\end{equation}
Set $z=\sqrt{T} \sigma(K,T)$. Using the formulae for the Black-Scholes Delta and Vega:   
\begin{equation}\label{BSchain}
\S(K, T) =  \frac{\Nn(-d(z)) -\QQ(S_T \geq K)}{\sqrt{T}n(d(z))}.
\end{equation}
where we recall that $\Nn(\cdot)$ is the standard Gaussian cumulative distribution function. Since, as~$x$ tends to infinity
\begin{equation}
\Nn(-x) =  \frac{n(x)}{x}\left(1 -\frac{1}{x^2} + \Oo\left(x^{-4}\right) \right)
\qquad\text{and}\qquad
\frac{1}{d(x)} =  \frac{2}{x}+ \frac{4 \log(K)}{x^3}   + \Oo\left(x^{-5}\right),
\end{equation}
 then
\begin{equation}
 \frac{\Nn(- d(z))}{\sqrt{T}n( d(z) )} = \frac{1}{\sqrt{T}  d(z)} \left(1 -\frac{1}{d(z)^2} + \Oo\left(d(z)^{-4}\right) \right) = \frac{2}{\sqrt{T}  z} \left(1 +\frac{2 \log(K)- 4}{z^2} + \Oo\left(z^{-4}\right) \right),
\end{equation}
and~\eqref{longSkew} follows by substituting $z$ and combining the above 
with~\eqref{BSchain}.
\end{proof}

\begin{oss} 
If $S^0 = \exp(X)$ is a martingale for some L\'evy process~$X$, from the proof of Proposition~\ref{EMM}, 
our models can be written as~$S^0_{T_t}$ for some time change~$T_t$, 
so that~$S^0_{T_t}$ converges to zero almost surely as~$t$ tends to infinity, 
provided we know this to hold for $S^0_t$. 
Such a property for exponential L\'evy models can be proved using fluctuations identities, 
since the assumption $\E[X_1]<0$ implies~\cite[VI.4, Exercise 3]{BertoinBook} 
that $X_t$ diverges to~$-\infty$. 
Because of Jensen inequality, a negative first moment is always the case for~$X_t$ when~$S^0$ is a martingale. 
Regarding the absolute continuity of the price process, 
this follows from the fact that the law of the involved processes are weak solutions of fractional Cauchy problems. 
These can be found using arguments analogous to~\cite[Examples 5.2-5.4]{MeerscahertOCTRW}.
\end{oss}

Part~(i) of  this lemma implies that the level and skew of the implied volatility are entangled: one cannot modify the leading order~$1/T$ of the skew decrease without postulating a zero or diverging asymptotic implied volatility level. In turn, a declining implied volatility can only be attained through a convergence rate of option prices distributions to the spot price slower than Gaussian, which is precisely the distinguishing feature of anomalous diffusions-based models.  Part~(ii) is an already known fact, originally observed in~\cite[Lemma~2]{gerholdST}, which highlights a very stringent relationship between the prices of Digital options and the small-time at-the-money skew.
It will be used later in Corollary~\ref{cor:shortSkewDRD}.

\begin{thm}\label{priceAsymptotics} 
As $T$ tends to infinity, we have the following asymptotic expansions
for the Call price $C(K, T)$, for any $K\geq 0$:
\begin{itemize}
\item[(i)] in the DRD model with $\beta \in (0,1]$, with the interpretation that~$L_t=t$ when $\beta=1$, 
there exist $C_1^\beta$ and $c_\beta>0$ such that
\begin{equation}\label{BTLCallAs}
C(K, T) = 1-  \mathbbm{1}_{\{\beta\neq 1\}}\frac{C_1^\beta } {\Gamma(1-\beta)}\frac{1}{T^{\beta}}\left[1+\Oo\left(\frac{1}{T}\right)\right] - \frac{ c^\beta}{\Gamma(\beta)} \frac{e^{- T \psi_X(\ii/2)}}{T ^{3/2- \beta}} \left[1+\Oo\left(\frac{1}{T}\right)\right];
\end{equation}
\item[(ii)] in the SL model with $\beta \in (0,1)$, there exists $C_2^\beta>0$ such that
\begin{equation}\label{SLCallAs}
C(K, T) = 1-   \frac{C^\beta_2 }{\Gamma(1-\beta)}\frac{1}{T^{\beta}}
\left(1+\Oo\left(\frac{1}{T}\right)\right).
\end{equation}
\end{itemize}
\end{thm}

\begin{proof}
Since we are under the assumptions of Theorem~\ref{pricingThm}, we can consider the price representation for a Call option
\begin{equation}\label{LewisCall}
C(K, T)= 1 - \frac{1}{2 \pi}\int_{-\infty}^{\infty} \frac{e^{\left(\ii u + \frac{1}{2}\right)\log(K)}}{u^2+1/4}\Phi_T\left(u + \frac{\ii}{2}\right) du
\end{equation}
which can be obtained from~\eqref{pricingLewis}  by moving the integration contour inside the strip $\Im(z)=1/2$ and applying the Residue Theorem  (see~\citealt{Lewis}).
Now the integrand in~\eqref{LewisCall} is bounded by an integrable function, thus by dominated convergence we   can take the limit as~$T$ tends to infinity of~$C(K,T)$ under the integral sign. So once we determined the asymptotic expansion of~$\Phi_T$ we can integrate the resulting expression to get the asymptotic equation of interest.

 Assume $\beta <1$ in the DRD model.  First of all, since the integration line contains points of variable argument, 
we must ensure that the Stokes phenomenon\footnote{The asymptotic behaviour of complex-valued functions can be different in different regions of the complex plane, which is normally referred to as Stokes phenomenon. A complex-valued function has a limit along a direction if it eventually takes values in only one of such areas.} does not occur. The asymptotic expansion of ${}_1F_1(a, b, z)$,  for large~$|z|$ is~\cite[Chapter 4]{Luke}:
\begin{equation}
{}_1F_1(a,b, z) \sim \frac{\Gamma(b)}{\Gamma(b-a)} z^{-a} e^{\ii \delta \pi a}
{}_2 F_0\left(a, 1+a-b, -\frac{1}{z}\right) + \frac{\Gamma(b)}{\Gamma(a)} z^{a-b} e^{z}
{}_2 F_0\left(b-a, 1-a, \frac{1}{z}\right)
\end{equation}
with $\delta=1$ if $\Im(z)>0$ and $\delta=-1$ otherwise.
So when $\Im(z)=1/2$,  since $\Re(\psi_X(z))>0$,  in equation~\eqref{CFS} for large~$T|\psi_X(z)|$ we have the well-defined behaviour 
\begin{align}\label{BTLexpansion}
{}_1F_1(\beta, 1, -T\psi_X(z)) \sim  & \frac{  (T\psi_X(z))^{-\beta}}{\Gamma(1-\beta)}
{}_2 F_0\left(\beta, \beta, (T\psi_X(z))^{-1}\right) \nonumber \\  
& + \frac{ e^{-T\psi_X(z)}  (- T\psi_X(z))^{\beta-1}  }{\Gamma(\beta)  }{}_2 F_0\left(1-\beta, 1-\beta, -(T\psi_X(z))^{-1}\right) \nonumber \\ \sim &\frac{  (T\psi_X(z))^{-\beta}}{\Gamma(1-\beta)} 
 + \frac{ e^{-T\psi_X(z)}  (- T\psi_X(z))^{\beta-1}  }{\Gamma(\beta)  }
\end{align}
 where in the last line we used that $\lim_{x\rightarrow 0}{}_2F_0(a,b;x)=1$, 
 for all $a,b$.
In order to substitute in~\eqref{LewisCall} the above expression for large $T$,  we need a uniformity argument in $u$.  Notice first that as~$|z|$ tends to infinity, $|\psi_X(z)|$ also tends to infinity 
because the risk-neutral drift of~$X$ must be nonzero by~\cite[Corollary 1.1.3]{BertoinSub}. 
This implicates that  along any line $\Im(z)=c$, $|\psi_X(z)|$  is strictly increasing. 
Also $\psi_X$ is even in its real part and odd in its imaginary part, so that $|\psi_X(\cdot+\ii/2)|$ on such sets must be an even function. 
We conclude that $|\psi_X(\cdot+\ii/2)|$,  has a positive minimum at the origin.  
Therefore  so long as $T$ is much larger than $1/|\psi_X(\ii/2)|$ we can replace $\Phi_T$ in~\eqref{LewisCall} with~\eqref{BTLexpansion}. Integrating the first term of the resulting expression produces the first term in~\eqref{BTLCallAs} with
\begin{equation}
C_1^\beta=\frac{1}{2 \pi}\int_{-\infty}^{\infty} \frac{e^{(\ii u + 1/2)\log(K)}}{(u^2+ 1/4) \psi_X(u+\ii/2)^{\beta}}du.
\end{equation}
Regarding the exponential sub-leading terms we have to analyse 
\begin{equation}
I^{\beta}(T):=
\int_{-\infty}^{\infty} \frac{e^{(\ii u + 1/2)\log(K)}e^{- T \psi_X(u+\ii/2)} }{(u^2+ 1/4)  \psi_X(u+\ii/2)^{1-\beta}}du,
\end{equation}
which can be treated using the saddle point method as in~\cite{AndersenLipton}. 
From the previous discussion,  $\psi_X(\cdot+ \ii/2)$  has a stationary point in 0 and further by $\psi''_X(\ii/2)>0$ by~\eqref{firstfinitemoment},  so that for large~$T$
\begin{equation}
I^{\beta}(T) \sim \frac{ \sqrt{2 \pi} 4 \sqrt{K}}{ \psi_X(\ii/2)^{1-\beta} \sqrt{\psi''_X(\ii/2)T}},
\end{equation}
which yields the second term in~\eqref{BTLCallAs} with
\begin{equation}
c_\beta= \frac{  4 \sqrt{K}}{ \psi_X(\ii/2)^{1-\beta} \sqrt{ 2 \pi \psi''_X(\ii/2)}}.
\end{equation}

When $\beta=1$ the whole proof collapses to the well-known steepest descent 
argument~\cite[Section 7]{AndersenLipton} for the L\'evy models price representation integral.

In the SL model we have, for any given $\beta<1$, that so long as $\pi \beta/2  < \theta < \min \{\pi, \pi \beta  \}$  the  asymptotic series for $E_{\beta}$
is given by~\cite[Equation 6.5]{HMS}
\begin{equation}\label{MLseries}
E_{\beta}(z) = \left\{\begin{array}{ll}
\displaystyle{ \frac{e^{z^{1/\beta}}}{\beta}\sum_{k=1}^{n-1} \frac{1}{\Gamma(1-\beta k)}\frac{1}{z^k} + \Oo\left(z ^{-n}\right)}, & \text{for  } |\arg(z)| <\theta,  \\      
\displaystyle{\sum_{k=1}^{n-1} \frac{1}{\Gamma(1-\beta k)}\frac{1}{z^{k}} + \Oo\left(z ^{-n}\right)}, & \text{for } \theta <|\arg(z)| \leq  \pi.
\end{array} \right.
\end{equation}
 Since   $\Re(\Psi_X(u+\ii/2))>0$,  for all $\alpha$  in the line $\Im(z)=1/2$ there exist $T_0$ big enough such that $\pi \beta < |\arg(- \psi_X(u+\ii/2) T_0^{\beta})|$, so that for $T>T_0$ the Stokes lines are not crossed. The correct expression is thus the  second line 
 in~\eqref{MLseries}, and we can repeat what argued in the DRD case.
 \end{proof}

\begin{oss}
The second term in~\eqref{BTLCallAs} is clearly negligible for large~$T$ compared to the leading order, when $\beta$ is smaller than~$1$. 
However for fixed $T$, as $\beta$ approaches one its contribution cannot be neglected. 
This term has been included to clarify the convergence to the L\'evy model. 
Such correction is not present in the SL model, and as $\beta=1$ the price approximation simply breaks down (however, by dominated convergence we still have convergence of prices).
\end{oss}

Proposition~\ref{priceAsymptotics}  clarifies the aforementioned slower convergence of Call prices compared to L\'evy (or exponentially-affine stochastic volatility) models. 
As already remarked, it can be thought of as a direct consequence of the slow, subdiffusive time spread of the asset returns. 
More specifically, the nature of the distribution implies that the pricing integral does not obey the Laplace decay rate, 
since the integrand is not of the form $\exp(-T f(x)) g(x)$. 
One instead obtains a vanishing long-term volatility, and hence by Lemma~\ref{asymptotics} 
a persistent long-term skew, as we illustrate below:

\begin{cor}\label{LambertLong} 
For $\beta \in (0,1)$, the leading-order asymptotic for large~$T$ of the  implied volatility 
in both the DRD and SL model satisfies
\begin{equation}\label{Lambert}
\sigma_\beta(K, T) \sim 2  \sqrt{ \frac{1 }{ T} W_0\left (\frac{2 K \, T^{2 \beta} \, \Gamma(1-\beta)^2 }{ \pi  C_{\beta}   }  \right)} ,
\end{equation}
where $W_0$ is the Lambert function and $C_\beta>0$. Furthermore, for all $K$, $\alpha > 1/2$,
\begin{equation}\label{skewconvergences}
\lim_{T \rightarrow \infty} \frac{T^{-\alpha}}{\S_\beta(K,T)}= 0
\qquad\text{and}\qquad
\lim_{T \rightarrow \infty} \frac{\S_\beta(K,T)}{\sqrt{T}}= 0.
\end{equation}
\end{cor}

\begin{proof}
Using $k=\log K$, the first-order expansion of the Black-Scholes price is simply
\begin{equation}
C_{\BS}(K, T, \sigma )=1 - 4 \sqrt{K}\frac{\exp\left(-\frac{\sigma^2 T}{8}\right)}{\sigma \sqrt{2\pi T}}\left(1+\Oo\left(\frac{1}{T}\right)\right).
\end{equation}
By definition of implied volatility, to determine the leading order of~$\sigma$ in the~DRD (respectively~SL) model, we need to equate the above respectively to the Call option price expansions~\eqref{BTLCallAs}
 (resp.~\eqref{SLCallAs}) and then solve for~$\sigma$. We have
\begin{equation}\label{ugly}
\exp\left(-\frac{\sigma^2  T}{8}\right)\frac{ 4 \sqrt{K}}{ \sigma \sqrt{2  T\pi} } =\frac{ C^2_{\beta} } {\Gamma(1-\beta)} T^{-\beta},
\end{equation}
with $C_\beta=C_i^\beta$, $i=1,2$ as in Proposition~\ref{priceAsymptotics}, depending on the model. Setting $ z=  \sigma^2 T/4$,  $M=\sqrt{2K}\Gamma(1-\beta)/ (C_\beta \sqrt{ \pi}) $, $w=M^2 \, T^{ 2 \beta}$,
then the equality~\eqref{ugly} reads $e^{ z } z =w$.
Since $w>0$ the inversion in $z$ can be performed along the real axis so that $W_0$ is well-defined,
and~\eqref{Lambert} follows.
Since $W_0(T)\sim \log(T)$ as~$T$ tends to infinity, then
\begin{equation}\label{asymptLev}
\sigma_\beta(K, T)\sim 2 \sqrt{ \frac{\log(M^2 T^{2 \beta})}{T}},
\end{equation}
therefore $T^{\alpha}\sigma_\beta(K, T)$ converges to zero for all $\alpha <  1/2$, 
which means that the first term of~\eqref{longSkew} tends to zero slower than $T^{-\alpha}$, for all $\alpha > 1/2$, but faster than $T^{-\alpha}$.

 
Studying the asymptotics of the last term in~\eqref{longSkew},
similar arguments to those of Proposition~\ref{priceAsymptotics}
imply that the long-term price decay for the Digital option $I_{ \{S_T \geq K \}}$  
is identical to that of the Call option, namely $c/T^\beta$ for some $c>0$. 
Then substituting~\eqref{asymptLev}  together with $d(x)\sim x/2$,
in the second term of~\eqref{BTLCallAs}, the proof follows from the asymptotic equivalence 
\begin{equation}
 \frac{\QQ(S_T \geq K)}{\sqrt{T}n(d(\sqrt{T} \sigma_\beta(K,T)) )} \sim  c \frac{\exp\left(\frac{T \sigma_\beta(K,T)^2}{8}  \right)}{T^{\beta+1/2}    } =
 c \frac{\exp\left( \frac{\log(M^{2} T^{ 2 \beta})}{2}\right)}{ T^{\beta+1/2}  }  = \frac{ c M }{\sqrt{T}}.
\end{equation}
\end{proof}

In light of the corollary above, persistence of the skew is to be interpreted as follows: 
the skew declines slower than any power of $T^{-1}$ bigger than $1/2$ (thus in particular, slower than~$1/T$) but always faster than $T^{-1/2}$.
It is then natural to ask if these structural differences in the implied volatility of anomalous diffusion models 
also affect the small-maturity limit. 
This turns out not to be the case, at least for the DRD model, and the underlying L\'evy model asymptotics are  maintained. 
More precisely, we have the following for Digital option prices:

\begin{prop}\label{shortProp} 
If the underlying L\'evy process~$X$ is such that 
\begin{equation}\label{seriesShort}
\QQ(S^0_t \geq 1 )= c_0 + c_{\eps} t^{\eps} + o(t^{\eps}),
\end{equation}
for some $c_0, c_\eps$, as~$t$ tends to zero, with $0<\eps \leq \frac{1}{2}$, then, 
with $c_{\beta, \eps}:=\frac{\Gamma(\beta+ \eps)}{\Gamma(\beta) \Gamma(1+\eps)}$,
\begin{equation}\label{digitalBTL}
\QQ\left(S^{DRD}_t \geq 1\right) = c_0 + c_{\beta, \eps} c_{\eps} t^{\eps} + o(t^{\eps}).
\end{equation}
\end{prop}

\begin{proof}
 Proposition~\ref{betaLaw} allows us to write 
\begin{equation}
\QQ\left(Y^{DRD}_t \geq 0\right)=\int_0^t \QQ(X_s \geq 0) \frac{s^{\beta-1}(t-s)^{-\beta}}{\Gamma(\beta) \Gamma(1-\beta)}ds.
\end{equation}
Now, notice that
$\E\left[\Bb_{\beta, 1-\beta}^a\right] = \frac{\Gamma(\beta+ a)}{\Gamma(\beta) \Gamma(1+a)}$
for all $a >0$, and that for $t$ sufficiently small,
\begin{equation}
\QQ(X_t \geq 0 ) = c_0 + c_{\eps} t^{\eps} + f(t),
\end{equation}
where $f(t) = o( t^{\eps})$ is a bounded function in a neighbourhood of the origin. The zero- and first-order terms of~\eqref{digitalBTL} are then clear, and by dominated convergence,
\begin{equation}
\lim_{t \rightarrow 0} \E[f(t \Bb_{\beta, 1-\beta})] t^{-\eps} =\int_0^1 \lim_{t \rightarrow 0} \frac{f(ts)}{t^{\eps}} \frac{s^{\beta-1}(1-s)^{-\beta}}{\Gamma(\beta) \Gamma(1-\beta)}ds  
=0,
\end{equation}
which yields the small-o order $\eps$ of the remainder.
\end{proof}

By combining this result with Lemma~\ref{asymptotics} we have the desired analogy of the short-term ATM skew of the DRD and L\'evy volatility asset pricing models. 

\begin{cor}\label{cor:shortSkewDRD}
Under the assumptions of Proposition~\ref{shortProp} the DRD model and its underlying L\'evy model have the same short-term ATM skew rate.
\end{cor}

\begin{proof}
Substituting respectively~\eqref{seriesShort} and~\eqref{digitalBTL} in~\eqref{shortSkew} yields the claim.
\end{proof}

As extensively discussed in \cite{gerholdST}, Equation~\eqref{seriesShort} essentially encompasses all the popular L\'evy models
and features very different behaviours: for example $c_0=1$ if the process has finite variation, whereas $c_0=1/2$ and $c_\eps=1/2+ d/(\sigma\sqrt{2\pi})$, $\eps=1/2$ 
for a jump diffusion with volatility~$\sigma$ and risk-neutral drift~$d$  \citep[Theorem~1 and Lemma~2]{gerholdST}. As it emerges, the critical value for which higher order terms are needed is $c_0=\frac{1}{2}$. In the DRD model, introducing $c_{\beta, \eps}$ does not change the asymptotic analysis, 
as~$c_0$ remains the same.


\begin{cor} 
If~$X$ satisfies~\eqref{seriesShort}, then the DRD model and the underlying exponential L\'evy model $S^0$ have the same short-maturity at-the-money skews.
\end{cor}


In the next section we bring together all these results and see how they lead to model calibration improvements when a persistent implied volatility skew is observed.

\section{Numerical analysis}\label{numerical}

\subsection{Volatility skew and term structure}
We visualize the volatility surfaces extracted from the DRD and SL models in Figures~\ref{surfaceBSSL} to~\ref{surfaceCGMYDRD}. 
For~$X$, we use a Brownian motion (Figures~\ref{surfaceBSSL} and~\ref{surfaceBSDRD}) 
and a CGMY process with parameters taken from~\cite{CGMY} (Figures~\ref{surfaceCGMYSL} and~\ref{surfaceCGMYDRD}), and consider moneynesses $\pm 40 \%$ ATM and maturities up to two years. 
In each figure, the smile of the anomalous diffusion is compared to that of its underlying L\'evy model~$S^0$.

First and foremost the slow decay of the volatility skew in anomalous diffusion models predicted by Proposition~\ref{priceAsymptotics} and 
Corollary~\ref{LambertLong} is apparent in all cases. 
Even though our result only predict an asymptotic rate of skew vanishing, 
our numerical tests, at least in the DRD case,  indicate that such a slower rate manifests itself already very early on. 
More research is necessary to see whether and how Proposition~\ref{priceAsymptotics} can be improved.
 
In Figures~\ref{surfaceBSSL} and~\ref{surfaceBSDRD} the volatility smile and skew 
of the anomalous diffusion model are present even if the L\'evy generating returns process is a Brownian motion. 
In other words, this confirms that introducing infinite-mean trade durations in a standard CTRW is alone  sufficient to generate a smile, consistently with Proposition~\ref{BTLmoments}. 
The smile appears rather symmetric, in line with the intuition that trade duration should have little skew impact, as it does not influence out-of-the-money prices any differently than in-the-money ones. 
This already suggests some orthogonality between~$\beta$ and the L\'evy parameters. 
In the Brownian motion case, $\beta$ is thus `overloaded', being responsible for both the smile convexity and its decay rate. This is relaxed in~\ref{surfaceCGMYSL} and Figures~\ref{surfaceCGMYDRD} by endowing~$X$ with a proper L\'evy structure (CGMY); there a short-term skew arises while the skew term structure maintains its slower flattening rate, dictated by~$\beta$. 
 
In Figures~\ref{surfaceBSSL} and~\ref{surfaceCGMYSL} we observe the repercussion on the implied volatilities of the `cross-over' phenomenon (Figure~\ref{functionals}) generated by the Mittag-Leffler and exponential types of the characteristics functions of the SL and pure L\'evy  models. The level of the SL surfaces transitions from a short-term regime where the implied volatilities are higher to a long-term one in which they are lower than those of the underlying L\'evy models (eventually tending to zero). Such a transition seems to be very sharp.

The time sections from Figures~\ref{surfaceCGMYSL} and~\ref{surfaceCGMYDRD} are shown in Figures~\ref{sectionsSLCGMY} and~\ref{sectionsDRDCGMY} and further highlight the remarks above. 
Figures~\ref{sectionsSLCGMYconv} to~\ref{sectionsDRDBSconv} highlight 
the convergence of the time sections to those of the underlying L\'evy model as~$\beta$ approaches one. 
For the SL model this convergence is from above, while it is from below for the DRD model. 
Note also that the DRD model exhibits a sharper ATM skew than the SL model.

\subsection{Calibration}

Corollary~\ref{LambertLong} and Proposition~\ref{shortProp} suggest that, from a calibration viewpoint, the models should behave as follows:
the L\'evy parameters have a short-time scale effect, unaffected when introducing~$\beta$, 
and they should hence absorb the short-time skew and smile. 
However, $\beta$ is the very component governing the long-term structure of the surface, 
where the L\'evy structure is flat and has no impact, and should thus allow to pick up the long-term skew. 
To test this we generate 3-month and 6-month volatility skews from a given L\'evy model $S^0$, 
which represent our baseline synthetic market data. In order to generate two scenarios of persistent volatility skew, while keeping the 3-month fixed, we shift the 6-month skew forward to make it coincide respectively with the 1-year and 18-month skews. 
We then cross-sectionally calibrate $S^0$, $S^{SL}$ and $S^{DRD}$ to the 3-month and 6-month skews in the baseline scenario, and the 3-month and 1-year (respectively 18-month) sections in the first and second scenarios. 

The calibration problem has been set-up as follows. 
Let $C(K,T ; \beta, \boldsymbol{\Gamma})$ be the theoretical Call prices from the SL or DRD models, where~$\boldsymbol{\Gamma}$ denotes the L\'evy parameters for~$X$, and by $C(K,T)$ the synthetic market prices obtained with the procedure described above. 
The goal is to minimize the root mean squared error (RMSE) on a set of quoted prices of~$n$ maturities and~$m$ strikes. The calibration parameters are then:

\begin{equation}\label{objective}
\argmin_{\beta, \boldsymbol{\Gamma}}\sqrt{\frac{1}{m \times n}\sum^{m,n}_{i,j=1} \left|C(K_i,T_j ; \beta, \boldsymbol{\Gamma})- C(K_i,T_j)\right|^2.}
\end{equation} 

We used $m=8$, corresponding to $\pm 20\%$ moneyness from $S_0=100$ and a step size of $\Delta K=5$. Further $n=2$ with $T_j$ chosen according to the scenarios illustrated above explained above.

To solve the optimization problem above we use the Differential Evolution global minimization algorithm. 
Its core MATLAB implementatio, from~\cite{MG}, is freely available at the GitLab repository \texttt{https://gitlab.com/NMOF/NMOF2-Code}. 
This code uses the Heston~\cite{Heston} characteristic function and its integral form for the Call price to be used  in the objective function~\eqref{objective}. The numerical quadrature method is Gauss-Legendre. 
We thus only need to replace it with our characteristic function~$\Phi_T$, 
which consists of the functions $E_{a}$ (for the SL model) and ${}_1F_1$ (for the DRD model), whose \texttt{.m}-files are available from the \texttt{Mathworks} website, 
and the characteristic exponent $\psi_X$ of the underlying driving returns L\'evy process~$X$. We make two distinct choices for~$X$: a Normal Inverse Gaussian (NIG;~\citealt{NIG}) process, with the parametrization from \cite{ContTankov}; and a Variance Gamma (VG;~\citealt{VG}) process. 
We therefore calibrate a total of four anomalous diffusions models, namely: SL-VG, SL-NIG, SL-VG and DRD-VG, plus the two underlying L\'evy models VG and NIG (using the obvious modification of~\eqref{objective}).

 The parameters of the Differential Evolution algorithm have been chosen as follows (for details see \citep{MG}). A total of $n_G=100$ number of generations (the halting condition) is produced, each with a population of $n_P=10d$ individuals spanning the search space, where $d$ is the number of parameters to be estimated. Therefore $n_P=40$ for both the DRD and SL models, and $n_P=30$ for the L\'evy models.  We used a mutation factor $F=0.5$ determining the deviation of the new population individuals from the current one, and a crossover probability $C_R=0.95$ of a mutation surviving before the selection test is applied. Further, every 10 generations we perform a Direct Search for solution improvement, each consisting of two searchers running a maximum of 100 iterations with a tolerance level of 0.001. 
This choice of parameters follows commonly applied rules of thumb which seem to work well for all models. The integral truncation for the Gauss-Legendre method is at $\pm 200$ and the integrand is evaluated at~$50$ nodes. 

 The results are shown in Tables~\ref{tablecomparNoshift} to~\ref{tablecompar18mths}. In Table~\ref{tablecomparNoshift} we represent the baseline scenario:  all three models perfectly fit the synthetic L\'evy market data. As expected, the $\beta$ parameter in the SL and DRD models calibrates to one, 
and produces no improvement on the $S^0$ calibration. 
In the scenarios with a persistent long-term skew, the total error for the L\'evy model is greater than that of both the SL and DRD models with $\beta<1$. 
Comparing among them the two scenarios we observe as expected that for both models, $\beta$ is smaller in the second scenario than in the first one, owing to a steeper long-term skew in the latter case. 
This can be interpreted as an asset with a more prolonged trade duration.

Comparing errors across the models, the SL model shows a better fit in all cases. However this should not 
necessarily be interpreted as an overall superiority: the better calibration might be only due to the synthetic market data generated by a L\'evy model, and the SL distributions being closer to L\'evy. 

\section{Conclusion}\label{conclusions}

We have proposed the use of anomalous diffusion processes in the context of option pricing, 
which allows to naturally incorporate trade durations between price moves. 
Using limits of CTRWs whose inter-arrival times distribution obeys a power law to model asset returns, 
we analysed the impact on the term structure of the returns distribution and on the corresponding implied volatility. 
More specifically, the observed volatility skew persistence on the market can be explained by a non-negligible impact of trade time randomness even in the long-term price evolution.

We analysed both cases when the price innovations are either dependent or independent from the waiting times between trades. Both models are consistent with the econometric observation that shorter duration generates sharper variations in the price revisions. Finally, we remarked  that even though the two models lead to similar large-maturity implied volatility properties, their different distributional properties produce rather different shapes of volatility surfaces. Numerical experiments confirm that for option pricing anomalous diffusions models have the potential to capture the slow decay of the volatility skew while retaining the short-term good properties of pure L\'evy models.

\bibliographystyle{apalike}
\bibliography{Bibliography}

\section{Tables and Figures}
\begin{figure}[ht!]
\centering \includegraphics[height=5cm, width=13cm]{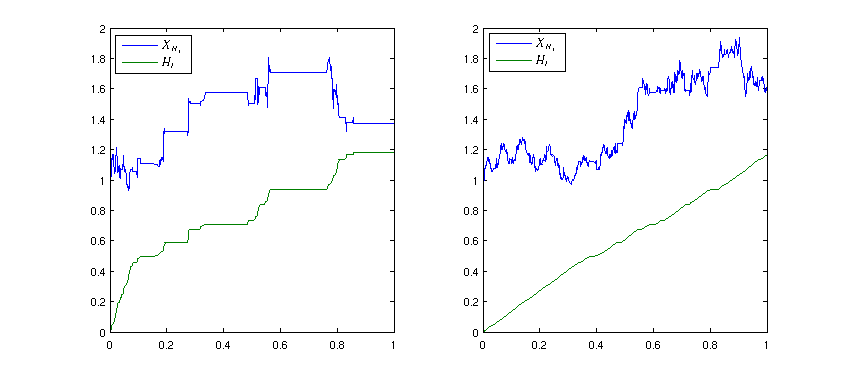}
\caption[]{Paths of $X_{H}$ (blue) and $H$ (green) in the SL model. 
$\beta=0.7$ on the left and $\beta=0.95$ on the right. 
Here,~$X$ is a driftless Brownian motion with diffusion parameter $\sigma=0.4$.}\label{plotsPaths}
\end{figure}

\begin{figure}[ht!]
\centering \includegraphics[scale=0.45]{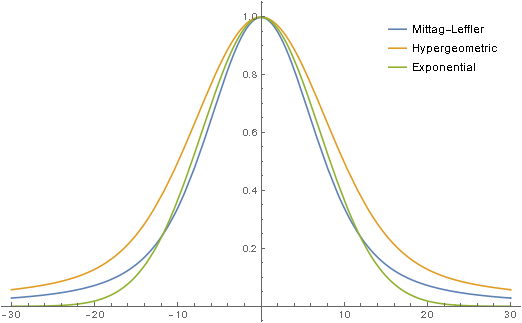}
\caption[]{Comparison of the function $\Phi_t$ for the SL and DRD model with the exponential,
with $\beta=0.75$, $t=0.5$. 
 We used the compensated geometric Brownian motion characteristic exponent 
 $\phi_X(z)= \sigma^2(z^2 - \ii z)/2$ along the line $\Im(z)=1/2$ where it is real.}\label{functionals}
\end{figure}

\begin{figure}[ht!]
\centering \includegraphics[scale=0.6]{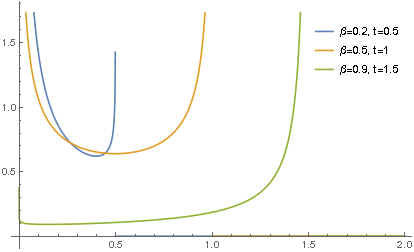}
\caption[]{Densities of the time change $L^H_t \sim t \Bb_{\beta, 1-\beta}$. 
For each $t$ the total integral at some value~$x$ has the interpretation of the probability that the time for the background L\'evy process~$X$ ran at most up to~$x$.
}\label{betas}
\end{figure}

\begin{figure}
\begin{floatrow}

\ffigbox{\includegraphics[scale=0.55]{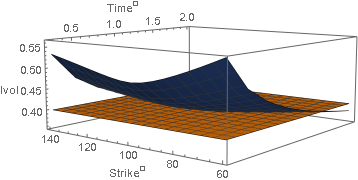}}
{\caption[]{SL implied volatility surface based on geometric Brownian motion;
$\sigma=0.4$, $\beta=0.7$.} \label{surfaceBSSL}
}
\ffigbox{\includegraphics[scale=0.55]{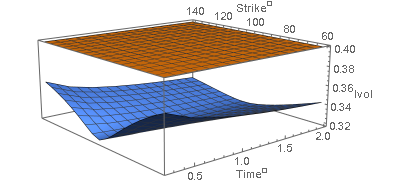}}
{\caption[]{DRD implied volatility surface from geometric Brownian motion;
$\sigma=0.4$, $\beta=0.7$.}\label{surfaceBSDRD}
}
\end{floatrow}
\end{figure} 

\begin{figure}
\begin{floatrow}

\ffigbox{\includegraphics[scale=0.55]{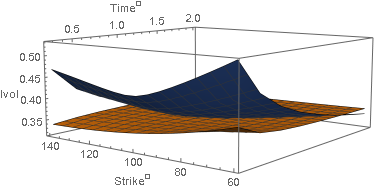}}
{\caption[]{SL implied volatility surface based on a CGMY L\'evy model, with $C=6.51, G=18.75, M=32.95, Y=0.5757, \beta=0.7$.} \label{surfaceCGMYSL}
}
\ffigbox{\includegraphics[scale=0.5]{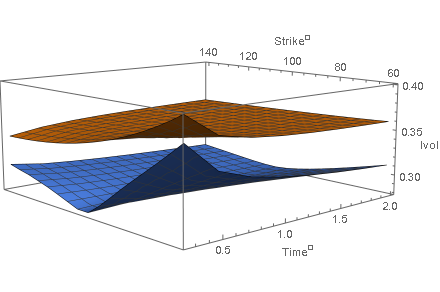}}
{\caption[]{DRD implied volatility surface based on a CGMY L\'evy model, with $C=6.51, G=18.75, M=32.95, Y=0.5757, \beta=0.7$.}\label{surfaceCGMYDRD}
}
\end{floatrow}
\end{figure} 

\begin{figure}
\begin{floatrow}

\ffigbox{\includegraphics[scale=0.4]{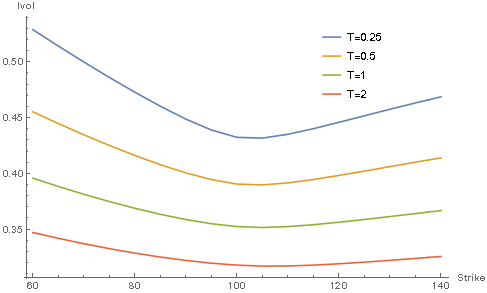}}
{\caption[]{Time sections from Figure~\ref{surfaceCGMYSL}.} \label{sectionsSLCGMY}
}
\ffigbox{\includegraphics[scale=0.4]{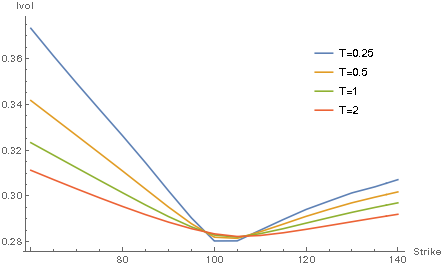}}
{\caption[]{Time sections from Figure~\ref{surfaceCGMYDRD}.}\label{sectionsDRDCGMY}
}
\end{floatrow}
\end{figure}

\begin{figure}
\begin{floatrow}

\ffigbox{\includegraphics[scale=0.45]{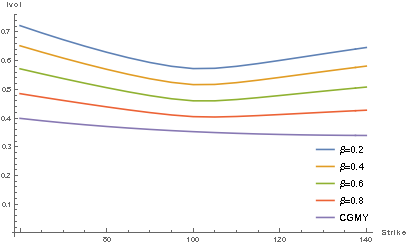}}
{\caption[]{Convergence of the SL skew to the CGMY one as $\beta$ tends to one, with $T=0.25$.} \label{sectionsSLCGMYconv}
}
\ffigbox{\includegraphics[scale=0.45]{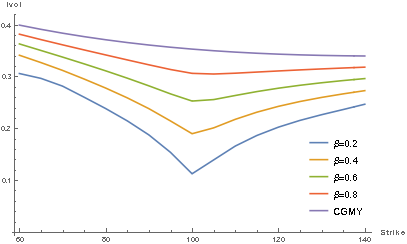}}
{\caption[]{Convergence of the DRD skew to the CGMY one as $\beta$ tends to one, with $T=0.25$.}\label{sectionsDRDCGMYconv}
}
\end{floatrow}
\end{figure}

\begin{figure}
\begin{floatrow}

\ffigbox{\includegraphics[scale=0.45]{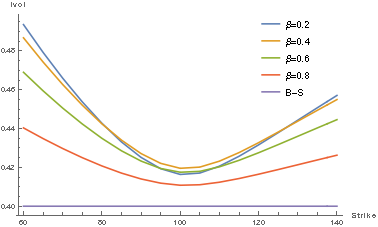}}
{\caption[]{Convergence of the SL model to the BS volatility as $\beta$ tends to one, for $T=0.75$.} \label{sectionsSLBSconv}
}
\ffigbox{\includegraphics[scale=0.45]{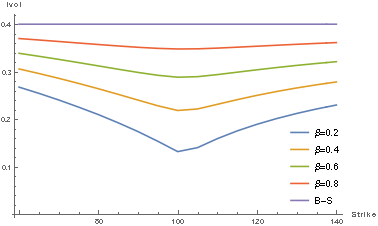}}
{\caption[]{Convergence of the DRD skew to the BS volatility as $\beta$ tends to one, for $T=0.75$.}\label{sectionsDRDBSconv}
}
\end{floatrow}
\end{figure}

\begin{table}[h!]

\begin{center}
\vspace{10pt}
\begin{tabular}{| c |  c | c |  c | c |  c | c |}
 \hline

\multirow{2}{*}{Parameter}  &    \multicolumn{2}{| c | }{ L\'evy } &    \multicolumn{2}{| c | }{ SL}  &  \multicolumn{2}{| c | }{ DRD }     \\  
\cline{2-7}

  & VG & NIG  & VG & NIG & VG &  NIG     \\  \hline
  
 $\kappa$  &  0.2037  &  0.2822  &  0.2037   & 0.2828   &  0.2037 &   0.2827  \\  \hline
  
$ \sigma$   & 0.3002  &  0.1994   &  0.3002  &  0.1989 &  0.3002  &   0.1994 \\  \hline

$\theta$   & -0.2983  &  -0.1039 &   -0.2984      & -0.1036  & -0.2984 &  -0.1038 \\  \hline
 
$\beta$    &  -          &  -                   &  1.0000    &     0.9977  &   0.9999   &    0.9999  \\ \hline
   
 RMSE   &  0.0084  &   0.0191        & 0.0084   &  0.0191    & 0.0084  &   0.0191 \\  \hline
\end{tabular}
\caption{Calibration to the 1-month L\'evy smile generated by the base model $S^0$. 
The parameters $(\kappa, \sigma, \theta)=(0.2,0.3,-0.3)$ for the VG model and  
$(0.3,0.2,-0.1)$ for the NIG model.}\label{tablecomparNoshift}
\end{center}

\end{table}

\begin{table}[h!]

\begin{center}
\vspace{10pt}
\begin{tabular}{| c |  c | c |  c | c |  c | c |}
 \hline
\multirow{2}{*}{Parameter}  &    \multicolumn{2}{| c | }{ L\'evy } &    \multicolumn{2}{| c | }{ SL}  &  \multicolumn{2}{| c | }{ DRD }     \\  

\cline{2-7}
  & VG & NIG  & VG & NIG & VG &  NIG     \\  \hline
  
 $\kappa$  &  1.4474  &  7.6080  &  1.5482   &  6.7626   &  0.9033  &   2.5647  \\  \hline
  
$ \sigma$   & 0.3298  &  0.2635   &  0.3218  &  0.2525 &  0.3758  &   0.3102 \\  \hline

$\theta$   & -0.1696  &  -0.0556&   -0.1739      & -0.0546  & -0.2824 &  -0.0995 \\  \hline
 
$\beta$    &  -             &  -        &  0.8669    &    0.8837   &    0.7224   & 0.6271     \\ \hline

 RMSE   &  0.3681    &   0.2061       & 0.2651    &  0.1729    & 0.2952   &   0.1791 \\  \hline

\end{tabular}
\caption{Calibration to the 1-month and 1-year shifted L\'evy smile generated 
by the base model~$S^0$. 
The parameters $(\kappa, \sigma, \theta)$ are $(0.2,0.3,-0.3)$ for VG and  
$(0.3,0.2,-0.1)$ for NIG.}\label{tablecompar12mths}
\end{center}
\end{table}

\begin{table}[h!]

\begin{center}
\vspace{10pt}
\begin{tabular}{| c |  c | c |  c | c |  c | c |}
 \hline
\multirow{2}{*}{Parameter}  &    \multicolumn{2}{| c | }{ L\'evy } &    \multicolumn{2}{| c | }{ SL}  &  \multicolumn{2}{| c | }{ DRD }     \\  
\cline{2-7}
  & VG & NIG  & VG & NIG & VG &  NIG     \\  \hline
 $\kappa$  &  4.5443  &  42.5059  &  3.2555   &  30.5836   &  2.0265 &    9.5124   \\  \hline
  
$ \sigma$   & 0.3952  &  0.4022   &  0.3661  &  0.3404 &  0.4628  &   0.4011 \\  \hline

$\theta$   & -0.1354  &  -0.0785  &   -0.1571      & -0.0711  & -0.2566 &  -0.1104 \\  \hline
 
$\beta$    &  -             &  -        &  0.8305   &    0.8634   &    0.6546 & 0.5704    \\ \hline
 
 RMSE   &  0.4857  &   0.2705        & 0.3612   &  0.2307   & 0.4157  &   0.2442\\  \hline
\end{tabular}
\caption{Calibration to 1-month and 18-month shifted L\'evy smiles 
generated by~$S^0$. 
The parameters $(\kappa, \sigma, \theta)$ are $(0.2, 0.3, -0.3)$ for the VG model and  
$(0.3, 0.2, -0.1)$ for the NIG model.}\label{tablecompar18mths}
\end{center}

\end{table}

\end{document}